\documentclass[12pt, draftclsnofoot, onecolumn]{IEEEtran}

\usepackage{cite}
\usepackage[cmex10]{amsmath}
\usepackage{amssymb}
\usepackage{amsthm}
\usepackage{mathrsfs}
\usepackage{bm}
\usepackage[mathscr]{eucal}
\usepackage{amssymb,amsmath,amsthm,amsfonts,latexsym}
\usepackage{amsmath,graphicx,bm,xcolor,url}
\usepackage{graphicx}
\usepackage{latexsym}
\usepackage{CJK}
\usepackage{indentfirst}
\usepackage{geometry}
\usepackage{psfrag}
\usepackage{setspace}
\usepackage{algorithmic}
\usepackage{algorithmic, cite}
\usepackage{algorithm}
\usepackage{array}
\usepackage{mdwmath}
\usepackage{mdwtab}
\usepackage{eqparbox}
\usepackage{url}
\usepackage{epstopdf}
\usepackage{epsfig,epsf,psfrag}
\usepackage{fixltx2e}
\usepackage{verbatim}
\usepackage{textcomp}
\usepackage{psfrag} 

\usepackage{subfigure} 
\usepackage{caption}

\theoremstyle{plain}
\newtheorem{mythe}{Theorem}
\theoremstyle{remark}

\theoremstyle{plain}

\theoremstyle{remark}

\theoremstyle{plain}

\theoremstyle{remark}

\theoremstyle{remark}

\theoremstyle{remark}

\theoremstyle{remark}

\theoremstyle{remark}

\theoremstyle{remark}

\catcode`~=11 \def\UrlSpecials{\do\~{\kern -.15em\lower .7ex\hbox{~}\kern .04em}} \catcode`~=13

\allowdisplaybreaks[4]

\newcommand{\calC}{\mathcal{C}}
\newcommand{\calD}{\mathcal{D}}

\newcommand{\calF}{\mathcal{F}}

\newcommand{\calN}{\mathcal{N}}
\newcommand{\calO}{\mathcal{O}}

\newcommand{\calR}{\mathcal{R}}

\newcommand{\calU}{\mathcal{U}}
\newcommand{\calV}{\mathcal{V}}

\newcommand{\ba}{\mathbf{a}}
\newcommand{\bA}{\mathbf{A}}

\newcommand{\boldf}{\mathbf{f}}

\newcommand{\bp}{\mathbf{p}}

\newcommand{\bv}{\mathbf{v}}

\newcommand{\bbE}{\mathbb{E}}

\DeclareMathAlphabet{\mathbsf}{OT1}{cmss}{bx}{n}
\DeclareMathAlphabet{\mathssf}{OT1}{cmss}{m}{sl}

\DeclareSymbolFont{bsfletters}{OT1}{cmss}{bx}{n}
\DeclareSymbolFont{ssfletters}{OT1}{cmss}{m}{n}
\DeclareMathSymbol{\bsfGamma}{0}{bsfletters}{'000}
\DeclareMathSymbol{\ssfGamma}{0}{ssfletters}{'000}
\DeclareMathSymbol{\bsfDelta}{0}{bsfletters}{'001}
\DeclareMathSymbol{\ssfDelta}{0}{ssfletters}{'001}
\DeclareMathSymbol{\bsfTheta}{0}{bsfletters}{'002}
\DeclareMathSymbol{\ssfTheta}{0}{ssfletters}{'002}
\DeclareMathSymbol{\bsfLambda}{0}{bsfletters}{'003}
\DeclareMathSymbol{\ssfLambda}{0}{ssfletters}{'003}
\DeclareMathSymbol{\bsfXi}{0}{bsfletters}{'004}
\DeclareMathSymbol{\ssfXi}{0}{ssfletters}{'004}
\DeclareMathSymbol{\bsfPi}{0}{bsfletters}{'005}
\DeclareMathSymbol{\ssfPi}{0}{ssfletters}{'005}
\DeclareMathSymbol{\bsfSigma}{0}{bsfletters}{'006}
\DeclareMathSymbol{\ssfSigma}{0}{ssfletters}{'006}
\DeclareMathSymbol{\bsfUpsilon}{0}{bsfletters}{'007}
\DeclareMathSymbol{\ssfUpsilon}{0}{ssfletters}{'007}
\DeclareMathSymbol{\bsfPhi}{0}{bsfletters}{'010}
\DeclareMathSymbol{\ssfPhi}{0}{ssfletters}{'010}
\DeclareMathSymbol{\bsfPsi}{0}{bsfletters}{'011}
\DeclareMathSymbol{\ssfPsi}{0}{ssfletters}{'011}
\DeclareMathSymbol{\bsfOmega}{0}{bsfletters}{'012}
\DeclareMathSymbol{\ssfOmega}{0}{ssfletters}{'012}

\newcommand{\bPhi}{\bm{\Phi}}

\def\norm#1{\left\| #1 \right\|}
\def\norm2#1{\left\| #1 \right\|_2}
\def\norm22#1{\left\| #1 \right\|_2^2}

\DeclareMathOperator{\diag}{diag}

\newcommand{\qednew}{\nobreak \ifvmode \relax \else
      \ifdim\lastskip<1.5em \hskip-\lastskip
      \hskip1.5em plus0em minus0.5em \fi \nobreak
      \vrule height0.75em width0.5em depth0.25em\fi}

\usepackage{cite}
\usepackage[cmex10]{amsmath}
\usepackage{amssymb}
\usepackage{amsthm}
\usepackage{mathrsfs}
\usepackage{bm}
\usepackage[mathscr]{eucal}
\usepackage{amssymb,amsmath,amsthm,amsfonts,latexsym}
\usepackage{amsmath,graphicx,bm,xcolor,url}
\usepackage{algorithmic}
\usepackage{graphicx}
\usepackage{diagbox}

\title{Reconfigurable Intelligent Surface Empowered Device-to-Device Communication Underlaying Cellular Networks}

\author{Gang~Yang, \emph{Member, IEEE}, Yating Liao, \emph{Student Member, IEEE}, Ying-Chang~Liang, \emph{Fellow, IEEE}, and Olav Tirkkonen, \emph{Member, IEEE} \\
\thanks{G.~Yang and Y. Liao  are with the National Key Laboratory of Science and Technology on Communications, and the Center for Intelligent Networking and Communications (CINC), University of Electronic Science and Technology of China (UESTC), Chengdu 611731, China (e-mails: yanggang@uestc.edu.cn, 2015010913035@std.uestc.edu.cn).}
\thanks{Y.-C. Liang is with the Center for Intelligent Networking and Communications (CINC), University of Electronic Science and Technology of China (UESTC), Chengdu 611731, China (e-mail: liangyc@ieee.org). (\emph{Corresponding author: Y.-C. Liang.})}
\thanks{O. Tirkkonen is with the Aalto University, Aalto, Finland. (e-mail:olav.tirkkonen@aalto.fi)}}

\begin{document}
\maketitle
\vspace{-1.5cm}
\begin{abstract}
Reconfigurable intelligent surface (RIS) is a  new and revolutionary technology to achieve spectrum-, energy- and cost-efficient wireless networks.  This paper studies the resource allocation for RIS-empowered device-to-device (D2D) communication underlaying a cellular network, in which an RIS is employed to enhance desired signals and suppress interference between paired D2D and cellular links. We maximize the overall network's spectrum efficiency (SE) (i.e., sum rate of D2D users and cellular users) and energy efficiency (EE), respectively, by jointly optimizing the resource reuse indicators, the transmit power and the RIS's passive beamforming, under the signal-to-interference-plus-noise ratio constraints and other practical constraints. To solve the non-convex problems, we first propose an efficient user-pairing scheme based on relative channel strength to determine the resource reuse indicators. Then, the transmit power and the RIS's passive beamforming are jointly optimized to maximize the SE by a proposed iterative algorithm, based on the techniques of alternating optimization, successive convex approximation, Lagrangian dual transform and quadratic transform. Moreover, the EE-maximization problem is solved by an alternating algorithm integrated with Dinkelbach's method. Also, the convergence and complexity of both algorithms are analyzed. Numerical results show that the proposed design achieves significant SE and EE enhancements compared to traditional underlaying D2D network without RIS and other benchmarks.
\end{abstract}
\section{Introduction}
\subsection{Motivation}
Device-to-device (D2D) communication underlaying a cellular network, which allows a device to communicate with its  proximity device over the licensed cellular bandwidth, is recognized as a promising wireless technology \cite{Jameel2018} and a competitive candidate for evolved/beyond 5th-Generation (5G) system standards\cite{Ghosh2019}. Specifically, the overall network's spectrum efficiency (SE) can be enhanced, since additional D2D links are supported by sharing the licensed cellular spectrums; the overall network's energy efficiency (EE) can be improved by exploiting the proximity of D2D users; also, the transmission delay can be reduced by eliminating the forwarding through a cellular base station (BS). However, interference management is one of the most important challenges for underlaying D2D communication\cite{Jameel2018,Ghosh2019,Naderializadeh2014}. The D2D link and the cellular link operating in the same licensed band interfere with each other severely\cite{DSABook2019Liang}, and the interference needs to be carefully suppressed via efficient interference control\cite{Sun2015D2D}\cite{Wang2018} and resource allocation\cite{Yang2016}\cite{Chen2018}. Existing interference management schemes were designed under the fact that wireless environment including interference channels is fixed. Thus the extent of interference suppression is fundamentally limited.


Recently, reconfigurable intelligent surface (RIS) has emerged as a new and revolutionary technology to achieve spectrum-, energy- and cost-efficient wireless networks\cite{liaskos2018new,di2019smart,Liang2019}. An RIS consists of a large number of passive low-cost reflecting elements, each of which can adjust the phase and amplitude of the incident electromagnetic wave in a software-defined way and reflect it passively\cite{huang2020holographic}. Thus, RIS is able to enhance desired signals and suppress interference by designing passive beamforming (i.e., changing each reflecting element's reflecting coefficient including amplitude and phase). In particular, a typical architecture of RIS consists of a smart controller and three layers (i.e., reflecting element, copper backplane, and control circuit board) \cite{liaskos2018new}. The controller attached to RIS can intelligently adjust the reflecting coefficients and communicate with other network components. Hence, it is realizable to intentionally reconfigure the wireless propagation environment and thus fundamentally improve the interference management level for underlaying D2D communication.

RIS can be explored to not only suppress the severe interference between each paired D2D link and cellular link, but also enhance the strength of desired signals for both D2D links and cellular links. This motivates us to study RIS-empowered D2D communication underlaying a cellular network as shown in Fig.~\ref{fig:system_newmodel}, which consists of multiple D2D pairs and multiple cellular users (CUs), as well as an RIS. This has not been studied in the literature to our best knowledge.

%
\subsection{Related Works}
\subsubsection{D2D Communication}
D2D communication systems have been widely studied in \cite{yu2011resource, feng2013device,Kai2019,MirzaJoint2018,cui2019spatial, Wang2015,Feng2015, Chun2017,WangD2D2019 ,Penda2019,Tang2017,Mozaffari2016}. Thereinto, interference management and performance analysis were investigated in \cite{yu2011resource, feng2013device,Kai2019,MirzaJoint2018,cui2019spatial, Wang2015,Feng2015, Chun2017}. For both direct D2D communication network and D2D communication underlaying cellular networks, the throughput over the shared resources was maximized in \cite{yu2011resource} by optimizing the resource allocation under the quality-of-service (QoS) requirements for both D2D users and CUs. For D2D communication underlaying cellular networks, the overall network's SE was maximized in \cite{feng2013device} by jointly optimizing the resource reuse indicators and transmit power. For D2D communication underlaying an orthogonal-frequency-division-multiplexing (OFDM) cellular network, the average ergodic sum rate over D2D pairs' locations was maximized in \cite{Kai2019} by jointly optimizing the subcarrier assignment and power allocation. For D2D communication underlaying a multiuser multiple-input multiple-output (MU-MIMO) cellular network, the total transmit power of the overall network was minimized in \cite{MirzaJoint2018}, by jointly optimizing the BS's transmit beamforming and the transmit power of both BS and D2D transmitters. For a direct D2D communication network, a deep learning approach was proposed in \cite{cui2019spatial} to maximize the overall network's SE by optimizing the scheduling of D2D links. The overall EE of an underlaying D2D network, which allows multiple D2D users pair with a CU, was maximized in \cite{Wang2015} by jointly optimizing the resource reuse indicators and power allocation. The overall EE was maximized in \cite{Feng2015} for dedicated transmission mode, reusing transmission mode and cellular transmission mode, while considering the circuit power consumption and the QoS requirements for D2D users and CUs. The performance of an underlaying D2D network over fading channels was analyzed in \cite{Chun2017} by leveraging a stochastic geometric approach.

D2D communication has also been incorporated with other advanced wireless communication technologies. For example, a two-phase cooperative transmission scheme was proposed in \cite{WangD2D2019} for D2D communication underlaying cellular networks, in which the policy of dynamic precoding and power allocation was designed. For  full-duplex D2D underlaying cellular networks, two cooperative modes based on network MU-MIMO and sequential forwarding, respectively, were proposed in \cite{Tang2017} to achieve both proximity gain and resource-reuse gain. The coverage and rate performance of UAV communication with underlaying D2D users were investigated in \cite{Mozaffari2016}.


\subsubsection{Wireless Communication with RIS}
Wireless communication systems with RIS can be divided into two categories, i.e., RIS-based transceiver design and RIS-assisted wireless communication. For the former, the RISs are utilized as transmit antennas and receive antennas to significantly reduce the hardware cost of traditional wireless transceivers\cite{Di2020smart}. Most literature were under the latter category, in which RIS is applied to improve the performance of wireless systems. RIS resembles but differs from existing technologies like full-duplex amplify-and-forward (AF) relaying and backscatter communications. A full-duplex AF relay actively processes the received signals and transmits the amplified signals, introducing additional noise and several interference at the relay; while an RIS passively reflects the signals instead of amplification, thus it avoids consuming power for active transmission and no noise or self-interference is introduced\cite{EmilIRSDFWCL2020}. Backscatter communication enables a tag to deliver its own information to a receiver by intentionally switching the antenna's load impedances\cite{YangLiangZhangPeiTCOM17}\cite{YangLiangZhangIoTJ18}, while an RIS is used to enhance the existing communication link performance.


RIS-assisted wireless communication was extensively studied in the prior works. For example, the weighted sum rate of an RIS-aided multiuser multiple-input single-output (MISO) downlink system was maximized in \cite{Guo2020}, by jointly optimizing the BS's active beamforming and the RIS's  passive beamforming (i.e., reflecting coefficients). The ergodic sum rate of an RIS-assisted MISO system was maximized in \cite{Huang2020} through deep reinforcement learning, by jointly optimizing the BS's transmit beamforming and the RIS's phase shifts. In \cite{Jung2019}, the distribution of SE was asymptotically analyzed and the reliability was verified for an RIS-empowered uplink system. For an RIS-assisted downlink non-orthogonal-multiple-access system, the max-min rate performance was optimized  in\cite{YangNOMA2020}. The EE of an RIS-empowered downlink multiuser system was maximized in \cite{8741198}, by jointly optimizing the BS's transmit power and the RIS's passive beamforming. The minimum secrecy rate of an RIS-assisted MISO system was maximized in \cite{JChen2019}.
\subsection{Contributions}
In this paper, we study the resource allocation for an RIS-empowered underlaying D2D communication network as shown in Fig.~\ref{fig:system_newmodel}. This work is an extension of the conference-version paper \cite{RISD2DsubmittedGC2020}. The main contributions are summarized as follows
\begin{itemize}
  \item We formulate a problem to maximize the overall network's SE (i.e., sum rate of both D2D users and CUs), by jointly optimizing the resource reuse indicators (i.e., user pairing between D2D users and CUs), the transmit power and the RIS's passive beamforming, subject to the signal-to-interference-plus-noise ratio (SINR) constraints for both D2D links and cellular links as well as other practical constraints. However, the problem is challenging to be solved optimally, since the user pairing (involving integer variables) and the resource allocation are closely coupled. 
  \item To decouple the SE-maximization problem, we first propose an efficient relative-channel-strength (RCS) based user-pairing scheme with low complexity. Under the obtained user-pairing design, an iterative algorithm based on alternating optimization (AO) is further proposed. Specifically, for given passive beamforming, the successive convex approximation (SCA) is exploited to optimize the transmit power; while for given transmit power, the Lagrangian dual transform and quadratic transform are exploited to solve the resulting multiple-ratio fractional programming problem. The algorithm's convergence and complexity are also analyzed.
  \item Based on the proposed model of RIS power consumption, we formulate a problem to maximize the overall network's EE, by jointly optimizing the resource reuse indicators, the transmit power and the RIS's passive beamforming, subject to the same constraints as the SE-maximization problem. To solve this non-convex problem, the proposed RCS based user-pairing scheme is first utilized to determine the resource reuse indicators, and an AO-based algorithm integrated with Dinkelbach's method is then proposed to optimize the transmit power and  RIS's passive beamforming iteratively. The algorithms's convergence and complexity are also analyzed.
  \item Numerical results show that the proposed design achieves significant SE and EE enhancements compared to traditional underlaying D2D without RIS, and suffers from slight degradation compared to the best-achievable performance under ideal user pairing. A 2-bit quantized phase shifter achieves sufficient SE enhancement compared to the ideal case of continuous phase shifter, and the highest EE for practical cases of finite-resolution phase shifter. The effects of main parameters on the performances are numerically verified, such as number of reflecting elements and CUs, and CUs' minimum rate requirement.
\end{itemize}

\subsection{Organization and Notations}
The rest of this paper is organized as follows. Section~\ref{systemmodel} presents the system model for RIS-empowered underlaying D2D communication network. Section~\ref{formulation} formulates the SE maximization problem. Section~\ref{solution} proposes an RCS based user-pairing scheme and an efficient iterative algorithm to solve the SE-maximization problem. Section~\ref{energyefficiency} formulates and solves the EE-maximization problem. Section~\ref{simulation} provides numerical results. Section~\ref{conslusion} concludes this paper.

The main notations are as follows. Denote scalars, vectors and matrices by italic letters, bold-face lower-case letters  and bold-face upper-case letters, respectively, e.g., $a$, $\ba$, $\bA$. Denote the space of  $x\times y$ complex matrices by $\calC^{x\times y}$.  Denote the set of real number and positive real number by $\calR$ and $\calR^{+}$, respectively. Denote the distribution of a circularly symmetric complex Gaussian (CSCG) random variable with mean $\mu$ and variance $\sigma^2$ by $\calC \calN (\mu,\sigma^2)$. Denote the transpose and conjugate transpose of a vector $\bv$ by $\bm v^T$ and $\bm v^H$, respectively. Denote the operation of taking real part by $\text{Re}\{\cdot\}$.
\vspace{-0.2cm}
\section{System Model} \label{systemmodel}
In this section, we first describe the RIS-empowered underlaying D2D communication network, and then present the signal model.
\subsection{System Description}
As shown in Fig.~\ref{fig:system_newmodel}, we consider an RIS-empowered cellular network with  underlay D2D, which consists of an RIS, $N \ (N \geq 1)$ D2D transmitters (TXs) denoted as $\text{TX 1},\ldots,\text{TX \emph{N}}$, $N$ D2D receivers (RXs) denoted as $\text{RX 1},\ldots,\text{RX \emph{N}}$, $K \ (K \geq 1)$ active CUs (i.e., cellular users) denoted as $\text{CU 1},\ldots,\text{CU \emph{K}}$, and a cellular BS. The RIS has $M \ (M \geq 1)$ reflecting elements, while each D2D TX, D2D RX, CU and the BS are equipped with a single antenna. A controller is attached to the RIS to control the reflecting coefficients and communicate with other network components through separate wireless links. We assume that the D2D links share the uplink (UL) spectrum of the cellular network, since the UL spectrum is typically underutilized compared to the downlink spectrum. To alleviate interference, we further assume that a D2D link shares at most one
CU's spectrum resource, while the resource of a CU can be shared  by at most one D2D link~\cite{feng2013device}~\cite{RamezanijointD2D}.

\begin{figure} [!t]
	\centering	\includegraphics[width=.5\columnwidth]{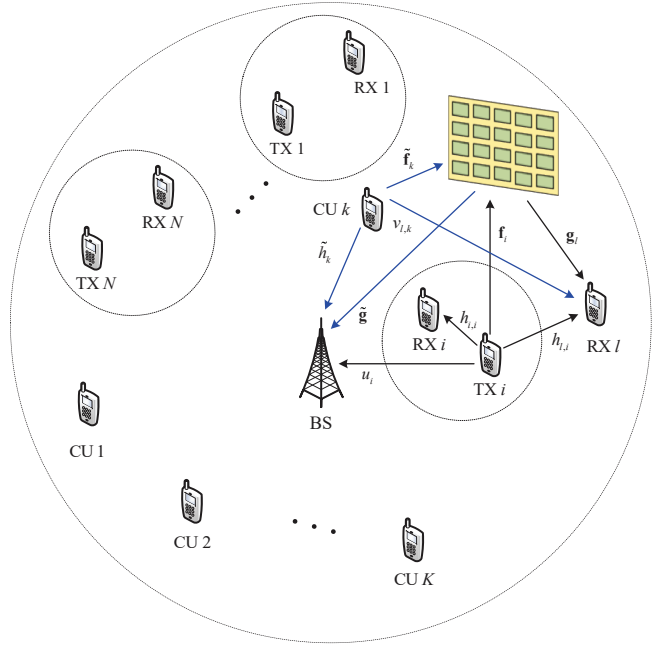}
	\caption{An  RIS-empowered underlaying D2D communication network.} \label{fig:system_newmodel}
\vspace{-0.4cm}
\end{figure}

All channels are assumed to experience quasi-static flat fading. The channels from $\text{TX \emph{i}} \ (1\leq i\leq N)$ to $\text{RX \emph{l}} \ (1\leq l\leq N)$ and RIS are denoted by $h_{l,i}\in \calC$ and ${\boldf }_i \in \calC^{M \times 1}$, respectively. For notational clarity, we represent each channel related to the cellular network  with a tilde. The channels from $\text{CU \emph{k}} \ (1\leq k\leq K)$ to BS and RIS are denoted by $ \tilde h_{k}\in \calC$ and $\tilde{\boldf }_k \in \calC^{M \times 1}$ , respectively; the channels from RIS to $\text{RX \emph{l}}$ and BS are denoted by $\mathbf{g}_l \in \calC^{M \times 1}$ and $\tilde{\mathbf{g}} \in \calC^{M \times 1}$, respectively; the interference channels from $\text{TX \emph{i}}$ to BS and from $\text{CU \emph{k}}$ to $\text{RX \emph{l}}$ are denoted by $u_i \in \calC$ and $v_{l,k} \in \calC$, respectively.

\subsection{Signal Model}
Let $\bm\Phi=\diag\{\alpha_1e^{j\theta_1},\ldots,\alpha_Me^{j\theta_M}\} \in \calC^{M\times M}$ denote the reflecting coefficient matrix of the RIS, where $\alpha_m\in\calR^+$ and $\theta_m\in\calR$ denote the reflecting amplitude and reflecting phase shift of the $m$-th reflecting element, for $1 \leq m \leq M$, respectively. Let $\beta_m=\alpha_me^{j\theta_m}\in\calF$ denote the reflecting coefficient, where $\calF$ is the feasible set of the reflecting coefficients. Three different settings for reflecting coefficients are considered in this paper.

\subsubsection{Ideal Reflecting Coefficients}
The amplitude and phase of each reflecting element are continuously adjustable. Specifically, the reflecting amplitude $\alpha_m \in [0,1]$, and the reflecting phase shift $\theta_m \in [0,2\pi)$. The set of all reflecting coefficients with ideal reflecting coefficient is
\begin{align}
&\calF_1= \left\{\beta_m=\alpha_me^{j\theta_m}\big||\beta_m|^2\leq1\right\}.
\end{align}

\subsubsection{Continuous Reflecting Phase Shift}
The reflecting phase shift $\theta_m$ takes continuous values in the range $[0,2\pi)$, while the reflecting amplitude $\alpha_m$ takes the maximum value of 1. The set of reflecting coefficients with continuous reflecting phase shift is
\begin{align}
&\calF_2= \left\{\beta_m=e^{j\theta_m}\big|\theta_m \in [0,2\pi)\right\}.
\end{align}

\subsubsection{Discrete Reflecting Phase Shift}
In this setting, $\alpha_m=1$, and the reflecting phase shift $\theta_m$ is $B$-bit quantized, taking $2^B$ discrete values. The corresponding set of reflecting coefficients is
\begin{align}
&\calF_3= \left\{\beta_m=e^{j\theta_m}\big|\theta_m \in \left\{0, \frac{2\pi}{2^B}, \ldots, \frac{2\pi(2^B-1)}{2^B}\right\}\right\}.
\end{align}

 From \cite{di2019smart}\cite{Liang2019}, different reflecting amplitudes and phase shifts can be realized by switching different resistor loads and setting different bias voltages to tuning elements that are typically varactor diodes, respectively. Due to hardware characteristic and cost limitations, it is practical to achieve finite-resolution phase shift. Therefore, the set of discrete reflecting phase shifts $\calF_3$ is in common use. Nevertheless, it is important to evaluate the system performances with $\calF_1$ and $\calF_2$, which provide upper bounds for the performance with $\calF_3$. 

The transmit signals from $\text{TX \emph{i}}$ and $\text{CU \emph{k}}$  are denoted as $s_i$ and $x_k$, respectively, which follow independent CSCG distribution with zero mean and unit variance, i.e., $s_i \sim \calC\calN (0,1)$, $x_k \sim \calC\calN (0,1)$. Denote the index set of active D2D pairs as $\calD \subseteq\{1,\ldots,N\}$. The corresponding SINR for $\text{RX \emph{n}}$ decoding $s_n$ from $\text{D2D TX \emph{n}}\in \calD$ is
\begin{align}
&\gamma_{n}^d=\frac{P_n^d\left|\mathbf{g}_{n}^H\bm \Phi {\boldf }_n+h_{n,n}\right|^2}{\sum\limits_{k=1}^{K} \rho_{k,n}P_k^c\left|\mathbf{g}_{n}^H\bm \Phi \tilde{\boldf }_k+v_{n,k}\right|^2+\sigma^2}, \label{eq:SINR_newd}
\end{align}
where $P_i^d$ and $P_k^c$ are the transmit power of $\text{TX \emph{i}}$ and $\text{CU \emph{k}}$, respectively; $\rho_{k,n}$ is the resource reuse indicator for cellular link $k$ and D2D link $n$, $\rho_{k,n}=1$ when D2D link $n$ reuses the resource of CU $k$, and $\rho_{k,n}=0$ otherwise; $\sigma^{2}$ is the power of additive white Gaussian noise (AWGN) at $\text{RX \emph{n}}$.


The SINR for the BS decoding $x_k$ from $\text{CU \emph{k}}$ is
\begin{align}
&\gamma_{k}^c=\frac{P_k^c\left|\tilde{\mathbf{g}}^H\bm \Phi \tilde{\boldf }_k+\tilde h_{k}\right|^2}{\sum\limits_{i=1}^{N} \rho_{k,i}P_i^d\left|\tilde{\mathbf{g}}^H\bm \Phi {\boldf }_i+u_{i}\right|^2+\sigma^2}, \label{eq:SINR_newc}
\end{align}
where $\sigma^{2}$ is the power of AWGN at the BS.

Hence, the overall network's SE (i.e., sum rate of both D2D users and CUs) in bps/Hz is
\begin{align}
&R\left(\bm\rho,\mathbf{p},\bm\Phi\right)=\sum \limits_{n\in \calD}\log_2(1+\gamma_n^d)+\sum\limits_{k=1}^{K}\log_2(1+\gamma_k^c),\label{eq:sumrate_new}
\end{align}
where the length-($KN$) resource reuse indicator vector $\bm\rho=\left[\rho_{1,1},\ldots,\rho_{1,N},\rho_{2,1},\ldots,\rho_{K,N}\right]^T$, and the length-($K+N$) power allocation vector  $\mathbf p=\left[P_1^d,\ldots,P_N^d,P_1^c,\ldots,P_K^c\right]^T$.

\section{PROBLEM FORMULATION FOR SE MAXIMIZATION}\label{formulation}
In this section, we formulate a problem to maximize the SE in \eqref{eq:sumrate_new}, by jointly optimizing the  resource reuse indicator vector $\bm{\rho}$, the transmit power vector $\mathbf p$ and the reflecting coefficients matrix $\bm\Phi$.  The optimization problem is formulated as follows
\begin{subequations}
\begin{align}
\text{(P1)}: \quad \underset{\bm\rho, \mathbf p, \bm \Phi}{\max}  \quad &R\left(\bm\rho, \mathbf p, \bm \Phi\right) \label{eq:P2objective}\\
\text{s.t.}\quad &\gamma_n^d\geq \gamma_{\min}^d ,  \quad n \in \calD \label{eq:P2constraint1_new} \\
&\gamma_k^c\geq \gamma_{\min}^c,\quad 1\leq k\leq K  \label{eq:P2constraint2_new} \\
& \sum_{k=1}^{K} \rho_{k,n}\leq 1 \label{eq:P2constraint3_new} \\
& \sum_{n\in\calD}^{N} \rho_{k,n}\leq 1 \label{eq:P2constraint4_new} \\
& 0\leq P_n^d \leq P_{\max}^d \quad  \label{eq:P2constraint5_new} \\
& 0\leq P_k^c \leq P_{\max}^c \quad \label{eq:P2constraint6_new} \\
& \beta_m \in \calF, \quad 1\leq m\leq M \label{eq:P2constraint7_new}
\end{align}
\end{subequations}
where  \eqref{eq:P2constraint1_new} and \eqref{eq:P2constraint2_new} indicate the required minimum SINRs (i.e., QoS) $\gamma_{\min}^d$ and $\gamma_{\min}^c$ for D2D links and cellular links, respectively; \eqref{eq:P2constraint3_new} ensures that a D2D link shares at most one CU's resource, while  \eqref{eq:P2constraint4_new} indicates that  the resource of a CU can be shared  by at most one D2D link; \eqref{eq:P2constraint5_new} and \eqref{eq:P2constraint6_new} are the maximum transmit power constraints on the TXs and CUs, respectively;  and \eqref{eq:P2constraint7_new} is the practical constraint on the reflecting coefficients with $\calF\in\{\calF_1, \calF_2, \calF_3\}$ .

Notice that (P1) is a non-convex problem. First, (P1) involves integer variables $\bm\rho$ and thus is NP-hard. Moreover, the objective function and the constraint functions of \eqref{eq:P2constraint1_new} and \eqref{eq:P2constraint2_new} are non-concave with respect to the variables $\bm {\rho}$, $\mathbf p$ and $\bm\Phi$, and these variables are all coupled. There is no standard method to solve such a non-convex problem. In the sequel, we first propose an user-pairing scheme with low complexity to determine the value of the resource reuse indicator vector $\bm\rho$. Then, we propose an efficient algorithm based on the AO (i.e., alternating optimization), SCA (i.e., successive convex approximation), Lagrangian dual transform  and  quadratic transform techniques to optimize $\mathbf p$ and $\bm\Phi$ in an iterative manner.

\section{Solution to SE-Maximization Problem} \label{solution} 
In order to solve the SE-Maximization problem (P1), we first propose an efficient user-pairing scheme to determine integer variables $\bm\rho$, then optimize $\mathbf p$ and $\bm\Phi$ in an iterative manner. To begin with, we solve (P1) with $\calF=\calF_1$, which makes  \eqref{eq:P2constraint7_new} a convex constraint. Therefore, the non-convexity of (P1) only roots from the objective function and other constraints. Afterwards, we  utilize the projection method to obtain heuristic solutions to (P1) with $\calF=\calF_2$ and $\calF=\calF_3$.

\subsection{Relative-Channel-Strength based Pairing Scheme}
Since the user-pairing design involves integer programming which is hard to solve, we propose a relative-channel-strength (RCS) based low-complexity pairing scheme to design the resource reuse indictors $\bm{\rho}$.

Notice that there are $A_K^N$ different possible pairings denoted as a set $\Pi\triangleq\{\pi_1,\ldots,\pi_{A_K^N}\}$. Each possible pairing can be viewed as an index mapping denoted as $\pi_q: k\in \calU_q {\longrightarrow} n \in\calD_q$, for $q=1,\ldots, A_K^N$, i.e., the $\pi_q$ maps each CU index $k \in \calU_q \subset \{1,2,\ldots,K\}$ to a D2D-link index $n \in \calD_q \subset \calD$. The RCS-based pairing scheme determines the pairing $\pi_{q^{\star}}$ by the following criterion
\begin{align}\label{eq:pairing}
\pi_{q^{\star}}=\underset{\pi_q \in \Pi}{\arg\max} \sum_{k \in \calU_q} \frac{|\tilde h_{k}|^2}{|v_{\pi_q(k),k}|^2}+\frac{| h_{\pi_q(k),\pi_q(k)}|^2}{|u_{\pi_q(k)}|^2}.
\end{align}

This heuristic pairing scheme chooses the pairing mapping which maximizes the sum of the relative channels that is defined as the ratio of (transmitter-to-receiver) useful channel strength over interference channel strength. Specifically, the first term in the summation of \eqref{eq:pairing} is the ratio of each paired CU-to-BS channel strength over the paired CU-to-RX interference channel strength, and the second term is the ratio of each paired TX-to-RX channel strength over the paired TX-to-BS interference channel strength.

Clearly, this heuristic pairing scheme that requires only simple comparison features low complexity, but fortunately its resultant design only suffers from slight performance degradation compared to the design with ideal pairing achieved by exhaustive search, as numerically shown in Section~\ref{simulation}. This RCS-based pairing scheme will also be used for EE maximization in Section \ref{energyefficiency}.

\subsection{Optimize Transmit Power Vector $\mathbf p$}
In each iteration $j$, for given reflecting coefficient matrix $\bm\Phi^{(j)}$, the transmit power vector $\mathbf{p}$  can be optimized by solving the following subproblem
\begin{subequations}
\label{eq:P1.1}
\begin{align}
\text{(P1.1)}: \quad \underset{\mathbf{p}}{\max} \quad &R(\mathbf{p}) \label{eq:P1.1objective}\\
\text{s.t.}\quad &\eqref{eq:P2constraint1_new}, \eqref{eq:P2constraint2_new}, \eqref{eq:P2constraint5_new}, \eqref{eq:P2constraint6_new}.
\end{align}
\end{subequations}

Since the objective function of (P1.1) is not concave with respect to the optimization variable $\mathbf p$, (P1.1) is non-convex. Notice that the objective function can be rewritten as follows
\begin{align}
R&=\sum \limits_{n\in \calD}\left[\log_2\left(P_n^dQ_{n,n}^{(j)}+A_1\right)-\log_2\left(A_1\right)\right] +\sum\limits_{k=1}^{K}\left[\log_2\left(P_k^cQ_{k}^{(j)}+A_2\right)-\log_2\left(A_2\right)\right], \label{eq:P1.1rewritten}
\end{align}
where $Q_{n,n}=|\mathbf{g}_{n}^H\bm \Phi {\boldf }_n+h_{n,n}|^2$, $\widetilde Q_{n,k}=|\mathbf{g}_{n}^H\bm \Phi \tilde{\boldf }_k+v_{n,k}|^2$, $\widetilde Q_{k}=|\tilde{\mathbf{g}}^H\bm \Phi \tilde{\boldf }_k+\tilde h_{k}|^2$, $Q_{i}=|\tilde{\mathbf{g}}^H\bm \Phi {\boldf }_i+u_{i}|^2$, $A_1=\sum_{k=1}^{K}\rho_{k,n}P_k^{c}\widetilde Q_{n,k}^{(j)}+\sigma^2$ and $A_2=\sum_{i=1}^{N} \rho_{k,i}P_i^d Q_i^{(j)}+\sigma^2$.

The non-convexity of \eqref{eq:P1.1rewritten} comes from the items $-\log_2(A_1)$ and $-\log_2(A_2)$. We exploit the SCA technique \cite{beck2010sequential} to solve (P1.1). Specifically, we need to find a concave lower bound to approximate the objective function. From the fact that any convex function can be lower bounded by its first-order Taylor expansion at any point, we obtain the following concave lower bound $R^{\text{lb}}$ at the local point $\mathbf{p}^{(j)}$
\begin{align}
R &\geq\sum \limits_{n\in\calD}\bigg[\log_2\left(P_n^d Q_{n,n}^{(j)}+A_1\right)\!\! -\log_2\left(A_1^{(j)}\right)  -\frac{1}{A_1^{(j)}} \sum\limits_{k=1}^{K}\rho_{k,n} \widetilde Q_{n,k}^{(j)}\left(P_k^c-P_k^{c(j)}\right)\bigg] \nonumber\\
&+ \sum\limits_{k=1}^{K}\bigg[\log_2\left(P_k^c\widetilde Q_{k}^{(j)}+A_2\right)-\log_2\left(A_2^{(j)}\right)
-\frac{1}{A_2^{(j)}}\sum\limits_{i=1}^{N} \rho_{k,i}Q_{i}^{(j)}\left(P_i^d-P_i^{d(j)}\right)\bigg]
\triangleq R^{\text{lb}}. \label{eq:taylor}
\end{align}

With given local point $\mathbf p^{(j)}$ and lower bound $R^{\text{lb}}$, the subproblem (P1.1) is approximated as
\begin{subequations}
\label{eq:P1.2}
\begin{align}
\text{(P1.1.A)}:\underset{\mathbf p}{\max} \quad &R^{\text{lb}}  \\
\text{s.t.}\quad &\eqref{eq:P2constraint1_new}, \eqref{eq:P2constraint2_new}, \eqref{eq:P2constraint5_new}, \eqref{eq:P2constraint6_new}.
\end{align}
\end{subequations}

Problem (P1.1.A) is a convex problem which can be efficiently solved with standard toolbox, e.g., CVX\cite{grant2008cvx}. Notice that the adopted lower bound $R^{\text{lb}}$  implies that the feasible set of (P1.1.A) is always a subset of that of  (P1.1). As a result, the optimal objective value obtained from  (P1.1.A) is in general a lower bound to that of  (P1.1).

\subsection{Optimize Reflecting Coefficient Matrix with $\calF=\calF_1$}\label{optimizePhi}
In each iteration $j$, for given transmit power vector $\mathbf p^{(j)}$, the reflecting coefficient matrix $\bm\Phi$ can be optimized by solving the following subproblem
\begin{subequations}
\label{eq:P2.1}
\begin{align}
\text{(P1.2)}:\underset{\bm \Phi}{\max} \quad &R(\bm \Phi) \label{eq:P2objective2.1}\\
\text{s.t.}\quad &\eqref{eq:P2constraint1_new}, \eqref{eq:P2constraint2_new}, \eqref{eq:P2constraint7_new}.
\end{align}
\end{subequations}

The logarithm in the objective function makes it difficult to solve (P1.2). Therefore, we tackle it via the Lagrangian dual transform proposed in \cite{shen2018fractional}. Introducing auxiliary variables $\boldsymbol\eta^d=[\eta_1^d,\ldots,\eta_N^d]^T$ and $\boldsymbol\eta^c=[\eta_1^c,\ldots,\eta_K^c]^T$, the subproblem (P1.2) can be equivalently reformulated as
\begin{subequations}
\begin{align}
\text{(P1.2.L)}:\quad\underset{\bm\Phi}{\max} \quad &R_a(\bm\Phi)   \label{eq:P2objective2.2}\\
\text{s.t.}\quad &\eqref{eq:P2constraint1_new}, \eqref{eq:P2constraint2_new}, \eqref{eq:P2constraint7_new},
\end{align}
\end{subequations}
where the new objective function $R_a(\bm\Phi)$ is expressed as
\begin{align}
R_a(\bm\Phi)&=\underset{\boldsymbol{\eta_n^d}}{\max} \left(\sum_{n\in\calD} \log\left(1\!+\!\eta_n^d\right)-\sum_{n\in\calD} \eta_n^d\!+\! \sum_{n\in\calD} \frac{(1+\eta_n^d)\gamma_n^d}{1+\gamma_n^d} \right)\nonumber\\
& \quad + \underset{\boldsymbol{\eta_k^c}}{\max} \left(\sum_{k=1}^{K} \log(1+\eta_k^c)-\sum_{k=1}^{K} \eta_k^c+ \sum_{k=1}^{K} \frac{(1+\eta_k^c)\gamma_k^c}{1+\gamma_k^c} \right). \label{eq:P1Obj}
\end{align}

Actually, we first optimize ${\eta_n^d}$ and ${\eta_k^c}$ with fixed ${\gamma_n^d}$ and ${\gamma_k^c}$, respectively; then optimize ${\gamma_n^d}$ and ${\gamma_k^c}$ with fixed ${\eta_n^d}$ and ${\eta_k^c}$, respectively. It can be easily checked that $R_a$ is a concave differentiable function over ${\eta_n^d}$ with ${\gamma_n^d}$ being fixed, so the optimal value of ${\eta_n^d}$ can be obtained by setting $\partial R_a(\gamma_n^{d(j)})/\partial {\eta_n^d}=0$, i.e., $ \eta_{\text{opt},n}^{d(j)}= \gamma_n^{d(j)}$. Similarly, $ \eta_{\text{opt},k}^{c(j)}= \gamma_k^{c(j)}$. Replacing ${\eta_n^d}$ and $\eta_k^c$ with $\eta_{\text{opt},n}^{d(j)}$ and $\eta_{\text{opt},k}^{c(j)}$, respectively, we find that the optimal objective values of (P1.2) and (P1.2.L) are equal, i.e., $R_a=R$.

Then, with given optimal $\boldsymbol{\eta_{\text{opt}}}^{d(j)}$ and $\boldsymbol{\eta_{\text{opt}}}^{c(j)}$, (P1.2.L) is transformed into
\begin{subequations}
\begin{align}
\text{(P1.2.E)}: \quad &\underset{\bm\Phi}{\max}\quad R_b(\bm\Phi) \\
\text{s.t.}\quad &\eqref{eq:P2constraint1_new}, \eqref{eq:P2constraint2_new}, \eqref{eq:P2constraint7_new},
\end{align}
\end{subequations}
where $R_b(\bm\Phi)$ is expressed as
\begin{align}
&R_b(\bm\Phi) =\sum_{n \in \calD} \frac{(1+\eta_{\text{opt},n}^{d(j)})\gamma_n^d}{1+\gamma_n^d}+\sum_{k=1}^{K} \frac{(1+\eta_{\text{opt},k}^{c(j)})\gamma_k^c}{1+\gamma_k^c}.
\end{align}

Further, we define $\bm\theta^H\bm{\omega}_{n,n}=\mathbf g_{n}^H\bm \Phi {\boldf }_n$, $\bm\theta^H\tilde{\bm{\omega}}_{n,k}=\mathbf g_{n}^H\bm \Phi \tilde{\boldf }_k$, $\bm\theta^H\tilde{\bm{\omega}}_{k}=\tilde{\mathbf{g}}^H\bm \Phi \tilde{\boldf }_k$ and $\bm\theta^H\bm{\omega}_{i}=\tilde{\mathbf{g}}^H\bm \Phi {\boldf }_i$. From \eqref{eq:SINR_newd} and \eqref{eq:SINR_newc}, optimizing the reflecting coefficient matrix $\bm\Phi$ can be equivalently transformed into optimizing $\bm\theta$ in the following objective function
\begin{align}
&R_b(\bm\theta)=\sum_{n\in\calD} \frac{\left(1+\eta_{\text{opt},n}^{d(j)}\right)P_n^{d(j)}Q_{n,n}^w}{P_n^{d(j)} Q_{n,n}^w+\sum\limits_{k=1}^{K}\rho_{k,n}P_k^{c(j)}\widetilde Q_{n,k}^w+\sigma^2}+\sum_{k=1}^{K} \frac{ \left(1+\eta_{\text{opt},k}^{c(j)}\right)P_k^{c(j)}\widetilde Q_{k}^w}{P_k^{c(j)}\widetilde Q_{k}^w+\sum\limits_{i=1}^{N} \rho_{k,i}P_i^{d(j)}Q_{i}^w+\sigma^2},\label{eq:Lagrangian}
\end{align}
where $Q_{n,n}^w=|\bm\theta^H\bm{\omega}_{n,n}+ h_{n,n}|^2$, $\widetilde Q_{n,k}^w=|\bm\theta^H\tilde{\bm{\omega}}_{n,k}+v_{n,k}|^2$, $\widetilde Q_{k}^w=|\bm\theta^H\widetilde{\bm{\omega}}_{k}+\tilde h_{k}|^2$ and $Q_{i}^w=|\bm\theta^H\bm{\omega}_{i}+ u_{i}|^2$.

Actually, (P1.2.E) can be equivalently reformulated as follows
\begin{subequations}
\begin{align}
\text{(P1.2.T)}:\underset{\bm\theta}{\max} \quad &R_b(\bm\theta)  \\
\text{s.t.}\quad &\eqref{eq:P2constraint1_new}, \eqref{eq:P2constraint2_new}, \eqref{eq:P2constraint7_new}.
\end{align}
\end{subequations}

Problem (P1.2.T) is a multiple-ratio fractional programming problem, which can be solved by utilizing the quadratic transform technique proposed in \cite{shen2018fractional}. Introducing auxiliary variable $\mathbf y=[y_1^d,\ldots,y_N^d,y_1^c,\ldots,y_K^c]^T$, the objective function of (P1.2.T) can be transformed as follows
\begin{align}
&R_b(\bm\theta, \mathbf y)=\nonumber\\
&\sum_{n\in\calD}\!\!\Bigg[\!2\sqrt{\!\! \left(1 \!+\!\eta_{\text{opt},n}^{d(j)}\right)P_n^{d(j)}} \text{Re} \left\{y_n^{d*}\bm\theta^H\bm{\omega}_{n,n}
\!+\! y_n^{d*} h_{n,n}\right\}\!-\!|y_n^d|^2\left(P_n^{d(j)} Q_{n,n}^w\!+\!\sum\limits_{k=1}^{K}\!\rho_{k,n}P_k^{c(j)}\widetilde Q_{n,k}^w+\sigma^2\right)\!\!\Bigg]\nonumber\\
&\!+\!\sum_{k=1}^{K}\! \Bigg[2\!\sqrt{(1 \!+\!\eta_{\text{opt},k}^{c(j)})P_k^{c(j)}} \text{Re} \left\{y_k^{c*}\bm\theta^H\widetilde{\bm{\omega}}_{k} \!+\! y_k^{c*}\tilde h_{k} \right\}\!-\!|y_k^c|^2\left(P_k^{c(j)}\widetilde Q_{k}^w\!+\!\sum\limits_{i=1}^{N} \!\rho_{k,i}P_i^{d(j)}Q_{i}^w\!+\!\sigma^2\right)\Bigg]. \label{eq:P2.4objectivetran}
\end{align}

Similarly, we first optimize $\mathbf y$ with fixed $\bm \theta$, then optimize $\bm \theta$ with fixed $\mathbf y$. It can be easily checked that $R_b(\bm \theta, \mathbf y)$ is a concave differentiable function over $\mathbf y$ with fixed $\bm \theta$, so the optimal solution of $\mathbf y$ can be obtained by setting $\partial R_b(\mathbf y, \bm \theta^{(j)})/ \partial {\mathbf y}=0$. Thus, the optimal value of $\mathbf y$ is given by
\begin{align}
&y_{\text{opt},n}^{{d(j)}}=\frac{\sqrt{ \left(1 + \eta_{\text{opt},n}^{d(j)}\right)P_n^{d(j)}} \left[\bm\theta^{H(j)}\bm{\omega}_{n,n} +  h_{n,n}\right]}{P_n^{d(j)} Q_{n,n}^w + \sum\limits_{k=1}^{K}\rho_{k,n}P_k^{c(j)}\widetilde Q_{n,k}^w + \sigma^2}, \label{eq:ynopt}\\
&y_{\text{opt},n}^{{c(j)}} = \frac{\sqrt{\left(1+ \eta_{\text{opt},k}^{c(j)}\right)P_k^{c(j)}} \left[\bm\theta^{H(j)}\widetilde{\bm{\omega}}_{k}+ \tilde h_{k}\right]}{P_k^{c(j)}\widetilde Q_{k}^w+\sum\limits_{i=1}^{N} \rho_{k,i}P_i^{d(j)}Q_{i}^w + \sigma^2}. \label{eq:ycopt}
\end{align}

Then, we optimize $\bm \theta$ for given $\mathbf y$. Replacing $y_n^{d{(j)}}$ and $y_k^{c{(j)}}$ with $y_{\text{opt},n}^{d{(j)}}$ and $y_{\text{opt},k}^{c{(j)}}$, respectively. Denote $\mathbf B_{1n}=P_n^{d(j)}\bm{\omega}_{n,n}\bm{\omega}_{n,n}^H$, $\mathbf  B_{2n}=\sum_{k=1}^{K}\rho_{k,n}P_k^{c(j)}\tilde{\bm{\omega}}_{n,k}\tilde{\bm{\omega}}_{n,k}^H$, $\mathbf B_{1b}=P_k^{c(j)}\tilde{\bm{\omega}}_{k}\tilde{\bm{\omega}}_{k}^H$, $\mathbf B_{2b}=\sum_{i=1}^{N}\!\!\rho_{k,i}P_i^{d(j)}\bm{\omega}_{i}\bm{\omega}_{i}^H$, $\mathbf e_{1n}=P_n^{d(j)} h_{n,n}^*\bm{\omega}_{n,n}$, $\mathbf e_{2n}=\sum_{k=1}^{K}\rho_{k,n}P_k^{c(j)}v_{n,k}^*\tilde{\bm{\omega}}_{n,k}$, $\mathbf e_{1b}=P_k^{c(j)}\tilde h_{k}^*\tilde{\bm{\omega}}_{k}$ and $\mathbf e_{2b}=\sum_{i=1}^{N}\!\rho_{k,i}P_i^{d(j)}u_{i}^*\bm{\omega}_{i}$.  Notice that the term $|\bm \theta^H \bm \omega+h|^2$ can be expanded as follows
\begin{align}
&|\bm \theta^H \bm \omega+h|^2= \bm \theta^H \bm \omega \bm \omega^H\bm \theta +2\text{Re}(\bm \theta^H \bm \omega h^*) +|h|^2. \label{eq:norm}
\end{align}

For given $\mathbf y$, the objective function is transformed as follows
\begin{align}
&R_b(\bm \theta, \mathbf y)=-\bm \theta^H\mathbf B_1\bm \theta+2\text{Re}\left(\bm \theta^H\mathbf e_1\right)+C_1,
\end{align}
where $C_1$ is a constant, the  matrix $\mathbf B_1$ and vector $\mathbf e_1$ are 
\begin{align}
&\mathbf B_1 = \sum_{n\in\calD} \left|y_{\text{opt},n}^{d(j)}\right|^2 \mathbf B_{n} +\sum_{k=1}^{K}\left|y_{\text{opt},k}^{c(j)}\right|^2 \mathbf B_{b} \\
&\mathbf e_1=\!\!\sum_{n\in\calD}\!\!\bigg[ \! \sqrt{ \! \left(1\!\! +\!\! \eta_{\text{opt},n}^{d(j)}\right)P_n^{d(j)}} \left(y_{\text{opt},n}^{d(j)}\right)^*\!\!\bm{\omega}_{n,n} \!\!-\!\! \left|y_{\text{opt},n}^{d(j)}\right|^2\!\!\mathbf e_{n} \!\! \bigg] \!\!+\!\!\sum_{k=1}^K\!\bigg[\!\! \sqrt{ \! \left(1\!\! +\!\! \eta_{\text{opt},k}^{c(j)} \!\right)P_k^{c(j)}} \left(\!y_{\text{opt},k}^{c(j)}\!\right)^*\!\widetilde{\bm{\omega}}_{k}\!\!-\!\!\left|y_{\text{opt},k}^{c(j)}\right|^2 \!\mathbf e_{b} \!\bigg],
\end{align}
with $\mathbf B_{n}=\mathbf B_{1n} +\mathbf B_{2n}$, $\mathbf B_{b}=\mathbf B_{1b} +\mathbf B_{2b}$, $\mathbf e_{n}=\mathbf e_{1n}+\mathbf e_{2n}$, and $\mathbf e_{b}=\mathbf e_{1b}+\mathbf e_{2b}$.

Similarly, introducing an auxiliary variable $x_d$, the constraint function of \eqref{eq:P2constraint1_new} can be equivalently written as
\begin{align}
f_d(\bm\theta,x_d)&= 2\sqrt{P_{n}^{d(j)}}\text{Re}\left(x_d^*\bm\theta^H\bm\omega_{n,n}+x_d^*h_{n,n}\right) -|x_d|^2\left(\sum\limits_{k=1}^{K}\rho_{k,n} P_k^{c(j)}\widetilde Q_{n,k}^w+\sigma^2\right). \label{eq:transform_quad}
\end{align}

Introducing an auxiliary variable $x_c$, the constraint function of  \eqref{eq:P2constraint2_new} can be equivalently written as follows
\begin{align}
f_c(\bm\theta,x_c)&= 2\sqrt{P_{k}^{c(j)}}\text{Re}\left(x_c^*\bm\theta^H\tilde{\bm\omega}_{k}+x_c^*\tilde h_{k}\right) -|x_c|^2\left(\sum\limits_{i=1}^{N}\rho_{k,i}P_i^{d(j)}Q_{i}^w+\sigma^2\right). \label{eq:transform_quac}
\end{align}

With $\bm \theta$ being fixed, $f_d(x_d,\bm\theta^{(j)})$ and $f_c(x_c,\bm\theta^{(j)})$ are concave differentiable functions over $x_d$ and  $x_c$, respectively.  The optimal solution of $x_d$ and  $x_c$ can be obtained by setting $\partial f_d(x_d,\bm \theta^{(j)})/ \partial {x_d}=0$ and $\partial f_c(x_c,\bm \theta^{(j)})/ \partial {x_c}=0$, respectively. The optimal values of $x_d$ and  $x_c$ are obtained, respectively, as follows
\begin{align}
&x_{\text{opt},d}^{(j)}= \frac{\sqrt{P_{n}^{d(j)}}(\bm \theta^{H(j)} \bm{\omega}_{n,n}+ h_{n,n})}
{\sum_{k=1}^{K} \rho_{k,n}P_k^{c(j)}\widetilde Q_{n,k}^w+\sigma^2}. \label{eq:xd_opt}\\
&x_{\text{opt},c}^{(j)}= \frac{\sqrt{P_{k}^{c(j)}}(\bm \theta^{H(j)} \tilde{\bm{\omega}}_{k}+\tilde h_{k})}
{\sum_{i=1}^{N} \rho_{k,i}P_i^{d(j)}Q_{b,i}^w+\sigma^2}. \label{eq:xc_opt}
\end{align}

Similarly, replacing $x_d^{(j)}$ and $x_c^{(j)}$ with $x_{\text{opt},d}^{(j)}$ and $x_{\text{opt},c}^{(j)}$, respectively, the following relationships are established
\begin{align}
&f_d(\bm\theta)=-\bm \theta^H\mathbf B_2\bm \theta+2\text{Re}\left(\bm \theta^H\mathbf e_2\right)+C_2  \geq \gamma_{\min}^d, \label{eq:con_trand}
\end{align}
\begin{align}
&f_c(\bm\theta)=-\bm \theta^H\mathbf B_3\bm \theta+2\text{Re}\left(\bm \theta^H\mathbf e_3\right)+C_3  \geq \gamma_{\min}^c, \label{eq:con_tranc}
\end{align}
where the positive-definite matrixes $\mathbf B_2$ and $\mathbf B_3$ are given by
\begin{align}
\mathbf B_2&=|x_{\text{opt},d}^{(j)}|^2\mathbf B_{2n}, \\
\mathbf B_3&=|x_{\text{opt},c}^{(j)}|^2\mathbf B_{2b},
\end{align}
the vectors $\mathbf e_2$ and $\mathbf e_3$ are given by
\begin{align}
\mathbf e_2&=\sqrt{P_{n}^{d(j)}}(x_{\text{opt},d}^{(j)})^*\bm{\omega}_{n,n}-|x_{\text{opt},d}^{(j)}|^2\mathbf e_{2n},\\
\mathbf e_3&=\sqrt{P_{k}^{c(j)}} (x_{\text{opt},c}^{(j)})^*\tilde{\bm{\omega}}_{k}-|x_{\text{opt},c}^{(j)}|^2\mathbf e_{2b},
\end{align}
and the constants $C_2$ and $C_3$ are given by
\begin{align}
C_2&=2\sqrt{P_{n}^{d(j)}} \text{Re}\left\{(x_{\text{opt},d}^{(j)})^* h_{n,n}\right\}-|x_{\text{opt},d}^{(j)}|^2\left(\sum_{k=1}^{K} \rho_{k,n}P_k^{c(j)}|v_{n,k}|^2+\sigma^2\right),\\
C_3&=2\sqrt{P_{k}^{c(j)}} \text{Re}\left\{(x_{\text{opt},c}^{(j)})^*\tilde h_{k}\right\}-|x_{\text{opt},c}^{(j)}|^2\left(\sum_{i=1}^{N}  \rho_{k,i}P_i^{d(j)}|u_{i}|^2+\sigma^2\right).
\end{align}

Therefore, (P1.2.T) is transformed as the following problem
\begin{subequations}
\begin{align}
\text{(P1.2.Q)}:\underset{\bm\theta}{\max}&  -\bm \theta^H\mathbf B_1\bm \theta+2\text{Re}(\bm \theta^H\mathbf e_1)+C_1  \\
\text{s.t.} &\quad \eqref{eq:con_trand}, \eqref{eq:con_tranc}, \eqref{eq:P2constraint7_new}.
\end{align}
\end{subequations}

The resulting (P1.2.Q) is a quadratic constrained quadratic programming (QCQP) problem, which can be effectively solved by standard toolbox, e.g., CVX\cite{grant2008cvx}.

\subsection{Optimize Reflecting Coefficient matrix $\bm\Phi$ with $\calF=\calF_2$}
In this subsection, we solve (P1) with $\calF=\calF_2$, which makes \eqref{eq:P2constraint7_new} a non-convex constraint. We utilize the projection method  to solve this non-convex problem. For the convenience of illustration, the optimal solutions to (P1) with $\calF=\calF_1$ and with $\calF=\calF_2$ are denoted by $(\mathbf p_{\text{opt},1}, \bm\Phi_{\text{opt},1})$ and $(\mathbf p_{\text{opt},2}, \bm\Phi_{\text{opt},2})$ respectively.

Notice that $\mathbf p_{\text{opt},2}$ take the values of $\mathbf p_{\text{opt},1}$, i.e., $\mathbf p_{\text{opt},2}=\mathbf p_{\text{opt},1}$. To obtain a suboptimal $\bm\Phi_{\text{opt},2}$, we project the solution to (P1) with $\calF=\calF_1$ into $\calF=\calF_2$. Specifically, we take the maximum value of the reflecting amplitude (i.e., $\alpha_m=1$), and the reflecting phases of $\bm\Phi_{\text{opt},2}$  take the same value of $\bm\Phi_{\text{opt},1}$. Therefore, the solution to (P1) with  $\calF=\calF_2$ can be written as
\begin{align}
& \bm\Phi_{\text{opt},2}= e^{j\arg(\bm\Phi_{\text{opt},1})}.
\end{align}

\subsection{Optimize Reflecting Coefficient matrix $\bm\Phi$ with $\calF=\calF_3$}
In this subsection, we solve (P1) with $\calF=\calF_3$, which makes (P1) a non-convex combinational optimization problem. Such a problem is indeed NP-hard, and its complexity of exhaustive-search method increases exponentially as the number of reflecting elements increases. Similarly, we exploit the projection method to solve this non-convexity. Denote the optimal solutions to (P1) with $\calF=\calF_3$  by $(\mathbf p_{\text{opt},3}, \bm\Phi_{\text{opt},3})$ .

We project the solution to (P1) with $\calF=\calF_1$ into $\calF=\calF_3$. Specifically, we take the maximum value of the reflecting amplitude (i.e., $\alpha_m=1$), and the reflecting phases of $\bm\Phi_{\text{opt},3}$  take the nearest value  of $\bm\Phi_{\text{opt},1}$. Therefore, the solution to (P1) with  $\calF=\calF_3$ can be written as
\begin{align}
& \theta_m = \underset{\theta_m \in \{0, \frac{2\pi}{2^B}, \ldots, \frac{2\pi(2^B-1)}{2^B}\} } {\arg\min} |\theta_m- \arg(\Phi_m^{\text{opt},1})|,
\end{align}
where $\Phi_m^{\text{opt},1}$ denotes the $m$-th element on the diagonal line of $\bm \Phi_{\text{opt},1}$. Similarly, $\mathbf p_{\text{opt},3}$ takes the value of $\mathbf p_{\text{opt},1}$, i.e., $\mathbf p_{\text{opt},3}=\mathbf p_{\text{opt},1}$.

\subsection{Overall Algorithm}
\begin{algorithm}[t!]
\caption{Proposed algorithm for solving (P1)}\label{AlgorithmP1}
\begin{algorithmic}
\STATE \textbf{Step 1}: Initialize $\bp^{(0)}, \ \bPhi^{(0)}$, a small threshold constant $\epsilon=10^{-2}$. Let $j=0$. \\
\STATE \textbf{Step 2}: Exploit RCS-based pairing scheme to determine the resource reuse indictor vector $\bm \rho^{\star}$.
\REPEAT
\STATE \textbf{Step 3}: Solve problem (P1.1.A) for given $\bm\Phi^{(j)}$, and obtain the optimal solution as $\bp^{(j+1)}$.
\STATE \textbf{Step 4}: Solve problem (P1.2.Q) for given $\bp^{(j+1)}$, and obtain the optimal solution as $\bm \Phi^{(j+1)}$.
\STATE  \textbf{Step 5}: Update iteration index $j=j+1$.
\UNTIL{The increase of objective value is smaller than $\epsilon$}.
\STATE  \textbf{Step 6}: Return the suboptimal solution  $\bm\rho^{\star}$, $\bp^{\star}=\bp^{(j-1)}$ and $\bPhi^{\star}=\bPhi^{(j-1)}$. 
\end{algorithmic}
\end{algorithm}

The overall algorithm is summarized in Algorithm \ref{AlgorithmP1}. There are three blocks of variables to be optimized, i.e., $\bm\rho$, $\mathbf p$ and $\bPhi$. We first use a low-complexity user-pairing scheme  based on the RCS to determine the resource reuse indicator vector $\bm\rho$.  Under the obtained user-pairing design, we use the AO algorithm to optimize $\mathbf p$ and $\bPhi$ alternatively in the out-layer iteration. Specifically, for given $\bm\Phi$, we utilize the SCA technique to tackle the non-convexity of the objective function; for given $\mathbf p$, we exploit the Lagrangian dual transform technique to deal with the sum of logarithm function, then utilize the quadratic transform technique to solve the resulting multiple-ratio fractional programming problem.

\subsection{Convergence and Complexity Analyses}
\subsubsection{Convergence Analysis}
The convergence of Algorithm \ref{AlgorithmP1} is given in the following theorem.

\begin{mythe}
Algorithm \ref{AlgorithmP1} is guaranteed to converge.
\end{mythe}

\begin{proof}
First, in Step 3, since the suboptimal solution $\mathbf p^{(j+1)}$ is obtained for given $\bPhi^{(j)}$, we have the following inequality on the sum rate 
\begin{align}
  R(\mathbf p^{(j)},\bPhi^{(j)})&\stackrel{(a)}{=}R^{\text{lb}}(\mathbf p^{(j)},\bPhi^{(j)}) \nonumber\\
  &\stackrel{(b)}{\le} R^{\text{lb}}(\mathbf p^{(j+1)},\bPhi^{(j)}) \nonumber\\
  &\stackrel{(c)}{=} R(\mathbf p^{(j+1)},\bPhi^{(j)}), \label{eq:convergence1}
\end{align}
where (a) and (c) hold since the Taylor expansion in \eqref{eq:taylor} is tight at given local point $\bp^{(j)}$ and $\bp^{(j+1)}$, respectively, and (b) comes from the fact that $\bp^{(j+1)}$ is the optimal solution to problem (P1.1.A).

Second, in Step 4, since $\bPhi^{(j+1)}$ is the optimal solution to problem (P1.2.Q), we can obtain the following inequality
\begin{align}
& R(\bp^{(j+1)},\bPhi^{(j)})\leq R(\bp^{(j+1)},\bPhi^{(j+1)}). \label{eq:convergence2}
\end{align}

From \eqref{eq:convergence1} and \eqref{eq:convergence2}, it is straightforward that
\begin{align}
  R(\bp^{(j)},\bPhi^{(j)})\leq R(\bp^{(j+1)},\bPhi^{(j+1)}),
\end{align}
which implies that the objective value of problem (P1) is non-decreasing after each iteration in Algorithm \ref{AlgorithmP1}. In addition, the objective value of problem (P1) is upper-bounded by some finite positive number since the objective function is continuous over the compact feasible set. Hence, the proposed Algorithm \ref{AlgorithmP1} is guaranteed to converge. This completes the convergence proof.
\end{proof}

\subsubsection{Complexity Analysis}
In Algorithm 1, the subproblems (P1.1.A) and (P1.2.Q) are alteratively solved in each outer-layer AO iteration, and the overall complexity of Algorithm 1 is mainly introduced by the update of the variables $\mathbf p$, $\bm\Phi$, $\mathbf y$, $x_d$ and $x_c$. Notice that the optimal $\mathbf y$, $x_d$ and $x_c$ are all obtained in closed forms, thus the computational complexity is negligible. Specifically, (P1.1.A)  can be solved in $\calO\left((N+K)^3\right)$ operations\cite{Bharadia2011}, while (P1.2.Q) is a convex QCQP which can be solved using interior point methods with complexity $\calO(M^{3.5})$\cite{hassanien2008robust}. Hence, the complexity of Algorithm \ref{AlgorithmP1} is $\calO\left(I_{\text{ite}}[(N+K)^3+M^{3.5}]\right)$, where $I_{\text{ite}}$ denotes the number of outer-layer AO iterations.
\section{ENERGY EFFICIENCY MAXIMIZATION}\label{energyefficiency}
In this section, we maximize the EE of the overall network, by jointly optimizing the  resource reuse indicator vector $\bm{\rho}$, the transmit power vector $\mathbf p$ and the reflecting coefficients matrix $\bm\Phi$.

\subsection{Problem Formulation for EE Maximization}
Before formulating the EE-maximization problem, we model the power consumption of the RIS. Due to the passive reflecting characteristic, the RIS's power consumption mainly comes from the control circuits\cite{8741198}. For typical control circuits, a field programmable gate array (FPGA) outputs digital control voltages with given sampling frequency, which are converted into analog control voltages by multiple digital-to-analog converters (DACs). The analog control voltage from each DAC adjusts the capacitance of each varactor diode in a continuous way, and thus controls the phase shift and amplitude of each element's reflected signals\cite{8741198}. Hence, the power consumption of RIS is modeled as follows
\begin{align}
P_{\sf RIS}(B)=P_{\sf FPGA}+MP_{\sf DAC}(B)+MP_{\sf v}(B),
\end{align}
where $P_{\sf FPGA}$, $P_{\sf DAC}(B)$ and $P_{\sf v}(B)$ denote the power of the FPGA, a $B$-bit DAC, and a varactor diode with $2^B$ different bias voltages, respectively. From \cite{Ribeiro2018}, the DAC's power $P_{\sf DAC}(B)=1.5\times10^{-5}\cdot2^{B}+9\times10^{-12}\cdot B\cdot f_s$, where $f_s$ is the sampling frequency. The varactor-diode power{\footnote{Notice that $P_{\sf v}(B)$ is negligible in practice, since the current of a varactor diode in the reversely-biased (until reverse breakdown) working status is almost constant and very small (typically, tens of nanoAmperes (nA)).}} $P_{\sf v}(B)=\bbE_{v_{\sf b} \in \calV}[v_{\sf b} i_{\sf r}]$, where $\calV = \{V_1, V_2, \ldots, V_{2^B}\}$ is the set of designed bias voltages.

Hence, the EE-maximization optimization problem is formulated as
\begin{subequations}
\begin{align}
\text{(P2)}:\quad&\underset{\bm\rho, \mathbf p, \bm \Phi}{\max}  \quad\frac{R\left(\bm\rho, \mathbf p, \bm \Phi\right)}{\sum\limits_{k=1}^KP_k^c+\sum\limits_{n\in\calD}P_n^d+(K+2N+1)P_0+P_{\sf RIS}(B)} \label{eq:P3objective}\\
&\text{s.t.}\quad\eqref{eq:P2constraint1_new},\eqref{eq:P2constraint2_new},\eqref{eq:P2constraint3_new},\eqref{eq:P2constraint4_new},\eqref{eq:P2constraint5_new},\eqref{eq:P2constraint6_new},\eqref{eq:P2constraint7_new},
\end{align}
\end{subequations}
where $\sum\limits_{k=1}^KP_k^c$ and $\sum\limits_{n\in\calD}P_n^d$ are the total transmit power of CUs and D2D transmitters, respectively, and $P_0$ is the circuit-power consumption at each transmitter or receiver of the overall network.

The constraints of (P2) which are the same as in (P1) are non-convex, and the objective function of (P2) is a fractional non-convex function with respect to the power allocation variables $\bp$. Hence, there is no standard method to solve (P2).


\subsection{Solution to (P2)}
To solve (P2), we first determine $\bm{\rho}$ through the RCS-based user-pairing scheme. Then, the variables $\mathbf p$ and $\bm\Phi$ are decoupled through AO technique.

\subsubsection{Solution to Subproblems}
In each iteration $j$, for given reflecting coefficient matrix $\bm\Phi^{(j)}$, the transmit power vector $\mathbf p$ can be optimized by solving the following subproblem
\begin{subequations}
\begin{align}
\text{(P2.1)}:\quad &\underset{\mathbf p}{\max}\quad \frac{R\left(\mathbf p\right)}{\sum\limits_{k=1}^KP_k^c+\sum\limits_{n\in\calD}P_n^d+(K+2N+1)P_0+P_{\sf RIS}(B)} \label{eq:P3.1objective}\\
&\text{s.t.} \quad\eqref{eq:P2constraint1_new},\eqref{eq:P2constraint2_new},\eqref{eq:P2constraint5_new},\eqref{eq:P2constraint6_new}.
\end{align}
\end{subequations}

We utilize the SCA technique as mentioned before to tackle the non-convexity of the  numerator in \eqref{eq:P3.1objective}, then transform it through the fractional programming into a parametric subtractive form with an introduced parameter $\lambda$, and exploit Dinkelbach's method \cite{crouzeix1991algorithms} to obtain a solution of $\lambda$ and $\mathbf p$. The solving sub-algorithm based on Dinkelbach's method is summarized in Algorithm \ref{Dinekelbach}.

\begin{algorithm}[t!]
\caption{Proposed algorithm for solving (P2)}\label{Dinekelbach}
\begin{algorithmic}
\STATE \textbf{Step 1}: Initialize $\bp^{(0)}, \ \bPhi^{(0)}$, $\lambda^{(0)}=0$, permissible error $\delta=10^{-3}$, a small threshold constant $\epsilon=10^{-2}$. Let $i=0$, $j=0$. \\
\STATE \textbf{Step 2}: Exploit RCS-based pairing scheme to determine the resource reuse indictor vector $\bm \rho^{\star}$.
\REPEAT
\STATE \textbf{Step 3}: Solve problem (P2.1) for given $\bm\Phi^{(j)}$, and obtain the optimal solution as $\bp^{(j+1)}$.
\REPEAT
\STATE 3.1:\quad Solve the following optimization problem to obtain the optimal transmit power $\mathbf p^{*(i)}$:$\underset{\mathbf p} \max \quad f(\lambda^{(i)})=R^{\text{lb(i)}}-\lambda^{(i)}(\sum\limits_{k=1}^KP_k^{c(i)}+\sum\limits_{n\in\calD}P_n^{d(i)}+(K+2N+1)P_0+P_{\sf RIS}(B))$
\STATE 3.2:\quad Update the introduced parameter with $\lambda^{(i)}=\frac{R^{\text{lb(i)}(\mathbf p^{*(i)})}}{\sum\limits_{k=1}^KP_k^{c*(i)}+\sum\limits_{n\in\calD}P_n^{d*(i)}+(K+2N+1)P_0+P_{\sf RIS}(B)}$
\STATE 3.3:\quad Update iteration index $i=i+1$.
\UNTIL{$f(\lambda)<\delta$}
\RETURN $\bp^{(j+1)}=\bp^{*(i-1)}$.
\STATE \textbf{Step 4}: Solve problem (P2.2) for given $\bp^{(j+1)}$, and obtain the optimal solution as $\bm \Phi^{(j+1)}$.
\STATE  \textbf{Step 5}: Update iteration index $j=j+1$.
\UNTIL{The increase of objective value is smaller than $\epsilon$.}
\STATE  \textbf{Step 6}: Return the suboptimal solution  $\bm\rho^{\star}$, $\bp^{\star}=\bp^{(j-1)}$ and $\bPhi^{\star}=\bPhi^{(j-1)}$. 
\end{algorithmic}
\end{algorithm}


For given transmit power vector $\mathbf p^{(j)}$, the reflecting coefficient matrix $\bm\Phi$ can be optimized by solving the following subproblem
\begin{subequations}
\begin{align}
\text{(P2.2)}: \quad&\underset{\bm \Phi}{\max}\frac{R\left(\bm \Phi\right)}{\sum\limits_{k=1}^KP_k^{c(j)}+\sum\limits_{n\in\calD}P_n^{d(j)}+(K+2N+1)P_0+P_{\sf RIS}(B)} \label{eq:P2objective}\\
&\text{s.t.} \quad \eqref{eq:P2constraint1_new},\eqref{eq:P2constraint2_new},\eqref{eq:P2constraint7_new}.
\end{align}
\end{subequations}

Since the denominator in the objective function is a constant, this subproblem (P2.2) can be solved in the same way as in \ref{optimizePhi} to obtain a solution to $\bm\Phi$.

\subsubsection{Overall Algorithm and Analyses}

The overall algorithm for solving (P2) is summarized in Algorithm 2. Specifically, the Dinkelbach-based Algorithm \ref{Dinekelbach} is used to solve subproblem (P2.1) in Step 3, and the subproblem (P2.2) is solved in Step 4. The subproblem (P2.1) and (P2.2) are alternatively solved in each outer-layer iteration.

The convergence of Algorithm \ref{Dinekelbach} can be proved by similar steps as in the proof of Theorem 1, thus omitted herein, since the Dinkelbach's method converges superlinearly for nonlinear fractional programming problems\cite{crouzeix1991algorithms}.

The proposed Dinkelbach's method solves a convex optimization problem in each iteration, thus the complexity of each iteration is $\calO\left(\log_2{(\frac{1}{\epsilon})}\right)$. The complexity for sovling (P2.1) with permissible error $\delta$ is $\calO\left(\frac{1}{\delta^2}\log_2{(\frac{1}{\epsilon})}\log_2(N+K)\right)$ \cite{crouzeix1991algorithms}.  (P2.2) can be transformed as a convex QCQP which can be solved using interior point methods with complexity $\calO(M^{3.5})$\cite{hassanien2008robust}. Hence, the complexity of Algorithm 2 is $\calO\left(I_{\text{ite}}\left[\frac{1}{\delta^2}\log_2{(\frac{1}{\epsilon})}\log_2(N+K)+M^{3.5}\right]\right)$, where $I_{\text{ite}}$ denotes the number of outer-layer AO iterations.

\vspace{-0.3cm}
\section{NUMERICAL RESULTS}\label{simulation}
This section provides numerical results for the RIS-empowered D2D underlaying cellular network, which show significant performance enhancement of the proposed design as compared to the conventional underlaying D2D network without RIS and other benchmarks.

\vspace{-0.2cm}
\subsection{Simulation Setups}
Each channel response consists of a large-scale fading component and a small-scale fading component. Without loss of generality, the large-scale fading is distance-dependent and can be modeled as $Cd^{-\alpha}$, where $d$ is the distance between transmitter and receiver with unit of meter (m), $C=10^{-3}$ is the path loss at the reference distance of 1 m, and $\alpha$ is the path loss exponent of the channel. The path loss exponents from TXs/CUs to RXs/BS are 4, from IRS to BS is 2, the others of RIS-related channels are 2.2\cite{feng2013device}\cite{8741198}. The small-scale fading components of  $h_{l,i}$, $u_i$, $\tilde h_{k}$ and $v_{l,k}$ are considered as independently Rayleigh fading distributed, while  the small-scale fading components of ${\boldf }_i$, $\tilde{\boldf }_k$, $\mathbf{g}_l$ and $\tilde{\mathbf{g}}$ follow independent Rician fading distribution, i.e.,
\begin{align}
&{\boldf }_i=\sqrt{\frac{K_1}{K_1+1}} {\boldf }_{\text{L},i} + \sqrt{\frac{1}{K_1+1}} {\boldf }_{N,i}, \label{eq:channel1}
\end{align}
where $K_1$ is the Rician factor of ${\boldf }_i$, ${\boldf }_{\text{L},i}$ is the line of sight (LoS) component, and ${\boldf }_{\text{N},i}$ is the non-LoS (NLOS) component where each element follows CSCG distribution $\calC \calN(0,1)$. Similarly, $\tilde{\boldf }_k$, $\mathbf{g}_l$ and $\tilde{\mathbf{g}}$ are generated in the same way as ${\boldf }_i$ with Rician factors $K_2$, $K_3$ and $K_4$, respectively. Rician factors are set as the same value, i.e., $K_1=K_2=K_3=K_4=10$.

\begin{figure}[t!]
\centering
\includegraphics[width=0.5\columnwidth]{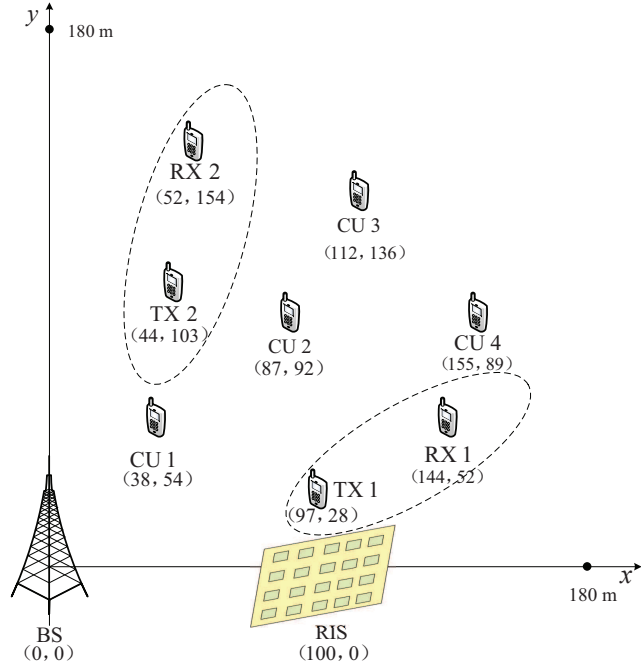}
\caption{Topology of RIS-empowered underlaying D2D communication network.} \label{fig:FigSim6}
\vspace{-0.1cm}
\end{figure}

We assume that the CUs are uniformly distributed in the cellular cell with radius $R=250$ m. We adopt the clustered distribution model in \cite{feng2013device} for D2D users, i.e., the clusters are randomly located in the cell, and each D2D link is uniformly distributed in one cluster with radius $r=60$ m. We set $K=4$ and $N=2$. The RIS is located between two D2D clusters. Using the above method, the locations of CUs and D2D users are generated by one realization, as illustrated in Fig. 2, and then fixed in all the simulations. The  coordinates of $\text{CU 1}$, $\text{CU 2}$, $\text{CU 3}$ and $\text{CU 4}$ are (38, 54), (87, 92), (112, 136) and (155, 89), respectively; the  coordinates of $\text{TX 1}$ and $\text{RX 1}$ are (97, 28), (144, 52), respectively; the  coordinates of $\text{TX 2}$ and $\text{RX 2}$ are (44, 103), (52, 154), respectively; the coordinate of  RIS is (100, 0). We set $\sigma^2=-114$dBm\cite{feng2013device}. Moreover, each antenna at the users is assumed to have an isotropic radiation pattern with 0 dB antenna gain, while each reflecting element of IRS is assumed to have 3 dB gain for fair comparison, since each IRS reflects signals only in its front half-space. The simulation results are based on 1000 channel realizations. Parameter settings are summarized in Table~\ref{table1}.
\begin{table}
\caption{Parameter settings.} \label{table1}
\footnotesize
\centering
\begin{tabular}{c|c}
  \hline
  Cellular cell radius & 250 m \\ \hline
  D2D cluster radius & $60$ m \\ \hline
  Number of reflecting elements & $M=200$ \\ \hline
  Number of CUs & $K=4$ \\ \hline
  Number of  D2D pairs & $N=2$ \\ \hline
  D2D TXs' maximum transmit power& $P_{\max}^d$=24 dBm \\ \hline
  CUs' maximum transmit power& $P_{\max}^c$=24 dBm \\ \hline
  D2D RXs' minimum SINR requirement& $R_{\min}^d$=0.3 bps/Hz \\ \hline
  CUs' minimum SINR requirement & $R_{\min}^c$=0.3 bps/Hz \\ \hline
  Noise power & $\sigma^2$=-114 dBm \\ \hline
  Channel realizations & 1000 \\ \hline
\end{tabular}
\end{table}



\subsection{Benchmark Schemes}
For comparison, we consider the following three benchmark schemes.
\subsubsection{Underlaying D2D Without RIS}
The traditional RIS-empowered D2D underlaying cellular network without RIS is considered. The SE and EE are maximized by jointly optimizing the resource reuse indicator $\bm{\rho}$ and transmit power vector $\mathbf{p}$. The solving algorithm in \cite{feng2013device} is used and omitted.
\subsubsection{Proposed Design with Random Reflecting Coefficients}
The proposed design without optimizing the reflecting coefficients matrix $\mathbf \Phi$ is considered. We maximize the SE and EE by jointly optimizing $\bm{\rho}$ and $\mathbf{p}$. All reflecting elements are set with random phase and maximal amplitude. This benchmark is used to show the benefit of passive beamforming optimization.
\subsubsection{RIS-Empowered D2D With Ideal User Pairing}
We exhaustively search over $A_K^N$ possible user-pairings, and jointly optimize $\bp$ as well as $\bPhi$ under each pairing. This benchmark gives achievable upper-bound performance of the RIS-empowered underlaying D2D network.


\subsection{Simulation Analyses for SE Maximization}
In this subsection, we evaluate the SE performance. We set $M=200$, $P_{\max}^d=P_{\max}^c=24$ dBm \cite{feng2013device}, $R_{\min}^d=\log_2 (1+\gamma_{\min}^d)=0.3$ bps/Hz, $R_{\min}^c=\log_2(1+\gamma_{\min}^c)=0.3$ bps/Hz, $P_{\max}^d=P_{\max}^c=P_{\max}$, if being not specified locally for some specific figure.

Fig.~\ref{fig:FigSim5} plots the SE versus the maximum transmit power $P_{\max}$ for the proposed design and different benchmarks. As shown in Fig.~\ref{fig:FigSim5}, the proposed design achieves significant SE enhancement compared to the first benchmark. For instance, the SE of the proposed design is $77.42\%$ and $73.8\%$ higher than that of the first benchmark, for $P_{\max}=20, 30$ dBm, respectively. Also, the proposed design outperforms the second benchmark without optimizing $\mathbf \Phi$, which shows the benefit of passive beamforming optimization. Compared to the third benchmark, the proposed design suffers from slight SE performance degradation, but obviously outperforms this benchmark in terms of computational complexity. The proposed design solves the joint-resource-allocation optimization problem only once, while this benchmark needs to solve such problem for $A_K^N$ times under all possible pairings, resulting into unaffordable complexity especially for large numbers of D2D or cellular links.

\begin{figure}[t!]
\begin{minipage}[t]{0.5\linewidth}
\centering
\includegraphics[width=1\textwidth]{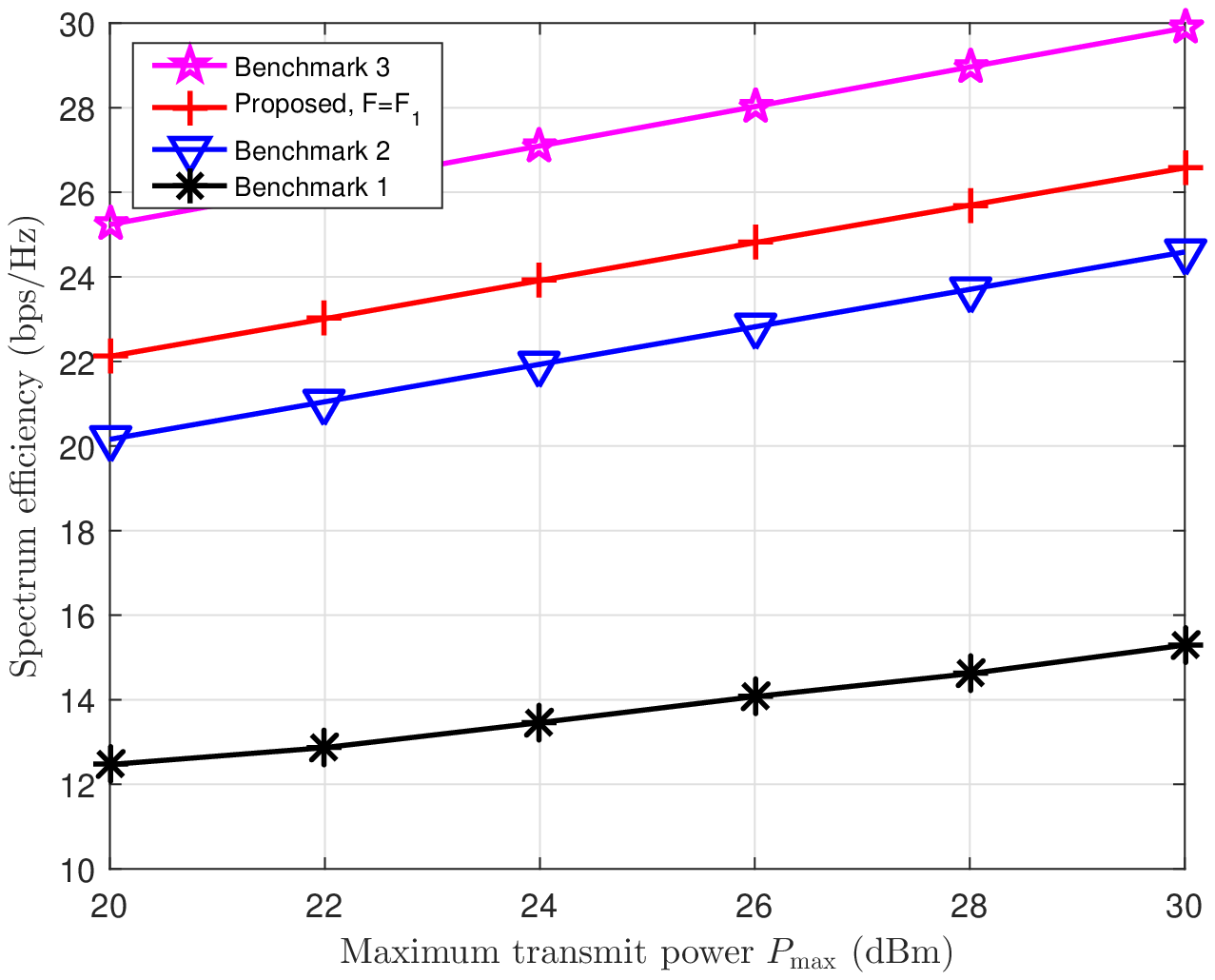}
\caption{SE versus $P_{\max}$  for proposed design and\\ different benchmarks.}
\label{fig:FigSim5}
\end{minipage}%
\begin{minipage}[t]{0.5\linewidth}
\centering
\includegraphics[width=1\textwidth]{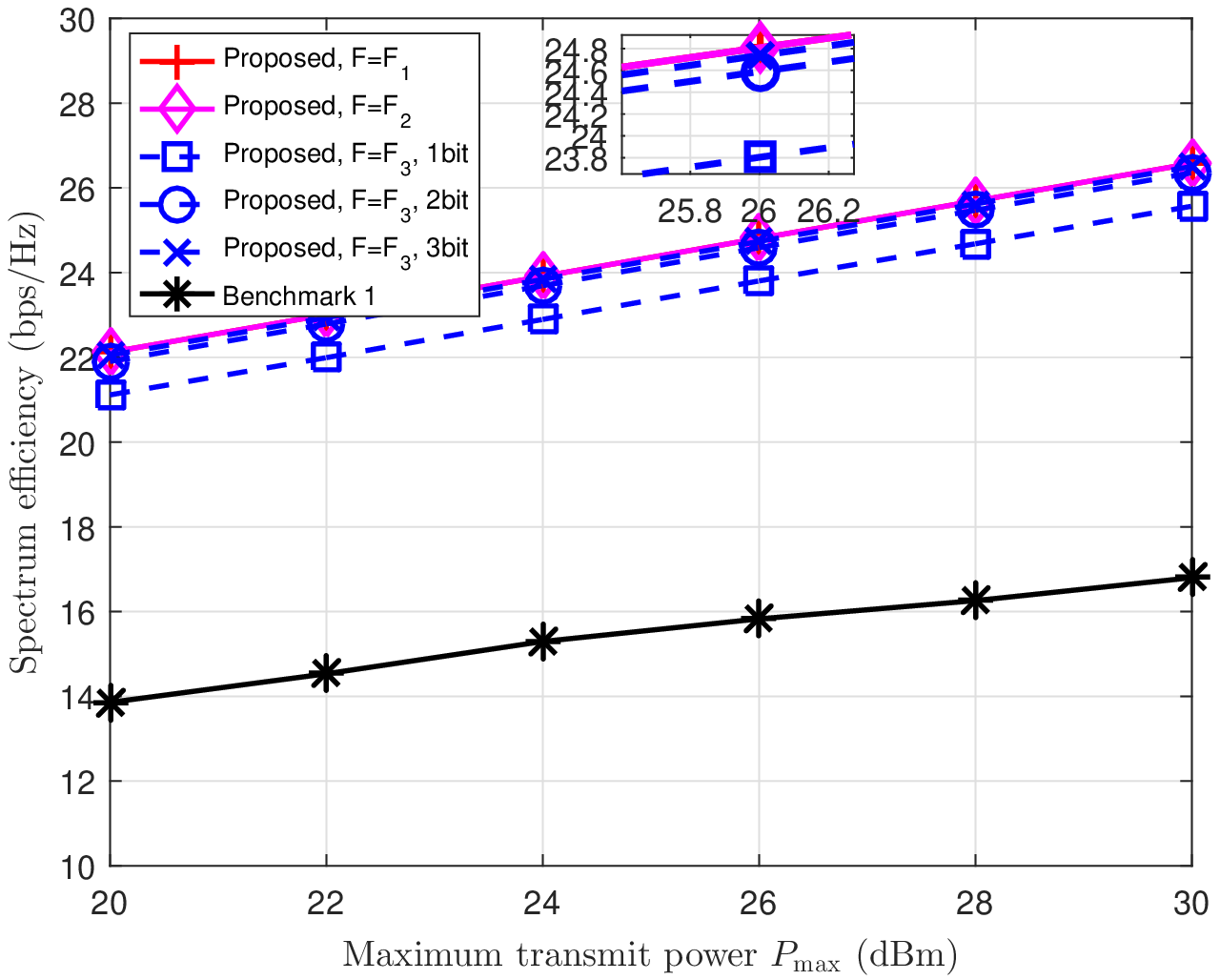}
\caption{SE versus $P_{\max}$ for proposed design with \\ different reflecting coefficients settings.}
\label{fig:FigSim6}
\end{minipage}
\vspace{-0.4cm}
\end{figure}

Fig.~\ref{fig:FigSim6} plots the SE versus the maximum transmit power $P_{\max}$ for the proposed design with different  reflecting coefficients settings. First, it is observed that the finite-resolution phase shifters of the reflecting elements usually degrades the SE performance. The SE increases with the increase of the phase-shift quantization bits $B$, since the increase of $B$ makes the setting of reflecting coefficients more accurate. In particular, the 2-bit phase shifter can obtain sufficiently high performance gain with a slight performance degradation compared to the ideal case of continuous phase shifters. Furthermore, the SE of the proposed design with $\calF=\calF_2$ is almost the same as the proposed design with $\calF=\calF_1$. The reason is as follows. As long as the reflecting amplitudes take the maximal value, channel strength enhancement and iter-link interference suppression can be achieved to the greatest extent by adjusting the reflecting phase shifts.

Fig.~\ref{fig:FigSim7} plots the SE versus the number of reflecting elements $M$ of the RIS. First, the SE of proposed designs increases as $M$ increases, since more reflecting elements can further enhance equivalent channel strength and suppress the inter-link interference; while the SE of first benchmark almost remains unchanged. Then, we observe that the gap between the proposed design and the second benchmark enlarges with the increase of  $M$, since more reflecting elements are well designed to achieve better performance for the proposed design, while the reflecting elements of the second benchmark stay  initial random values. Moreover, compared to the third benchmark with extremely high complexity, the proposed design achieves 79.67\% and 88.94\% SE performance of the third benchmark (upper bound) when $M$ is 50 and 400, respectively.


\begin{figure}[t!]
\begin{minipage}[t]{0.5\linewidth}
\centering
\includegraphics[width=1\textwidth]{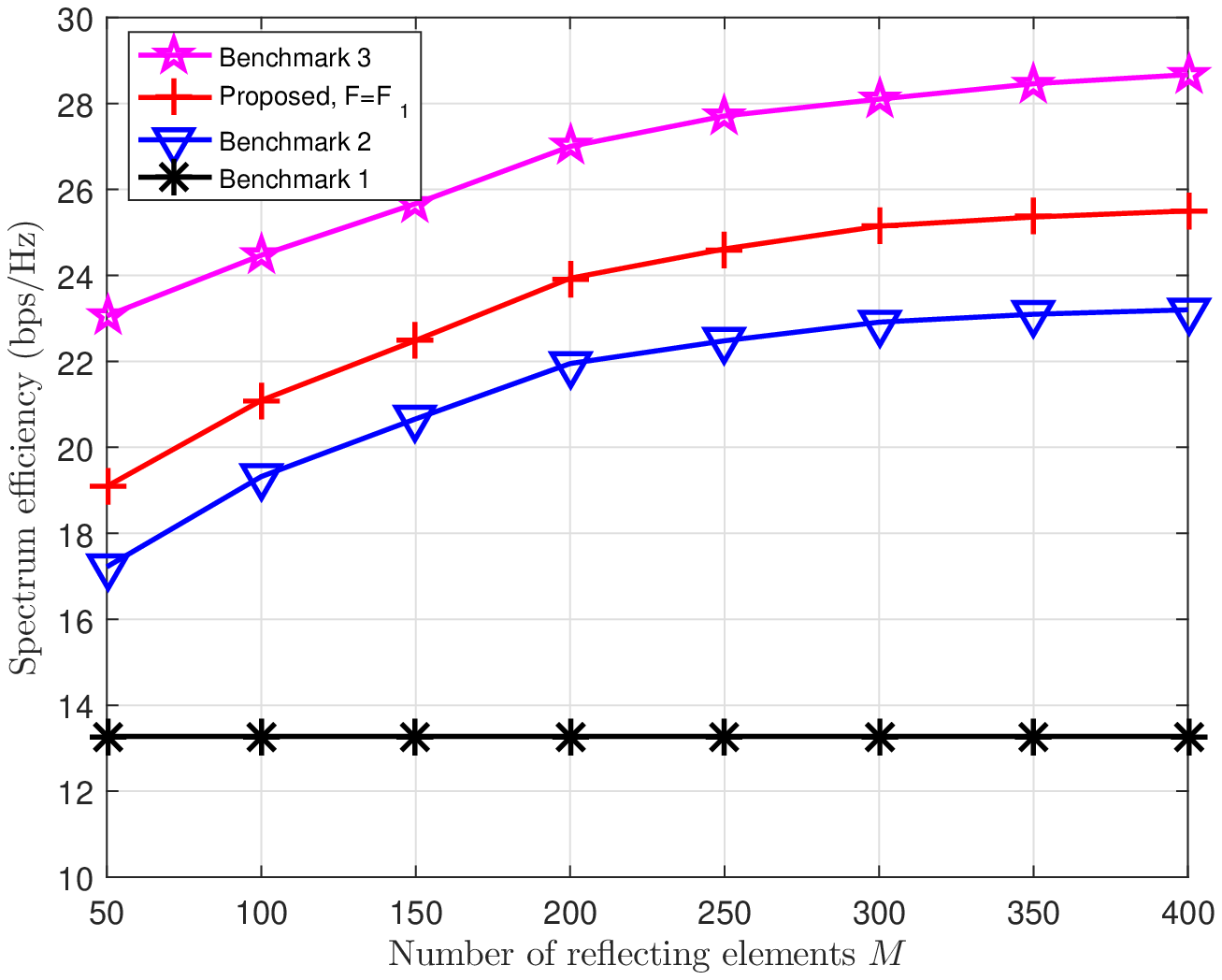}
\caption{SE versus RIS's number of reflecting \\elements $M$.}
\label{fig:FigSim7}
\end{minipage}%
\begin{minipage}[t]{0.5\linewidth}
\centering
\includegraphics[width=1\textwidth]{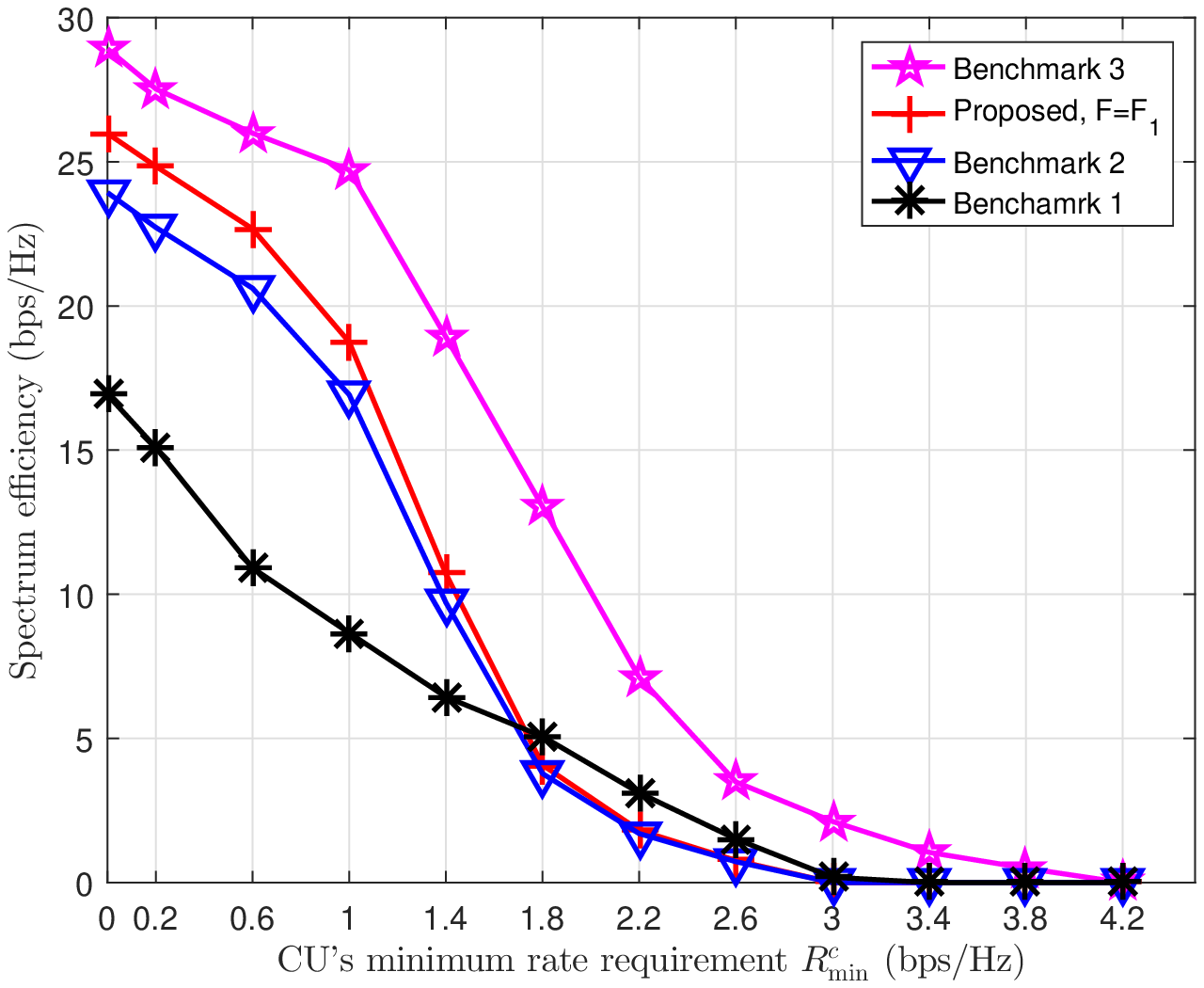}
\caption{SE versus CU's minimum rate require\\-ment $R_{\min}^c$.}
\label{fig:FigSim8}
\end{minipage}
\vspace{-0.4cm}
\end{figure}

\begin{figure}[t!]
\begin{minipage}[t]{0.5\linewidth}
\centering
\includegraphics[width=1\textwidth]{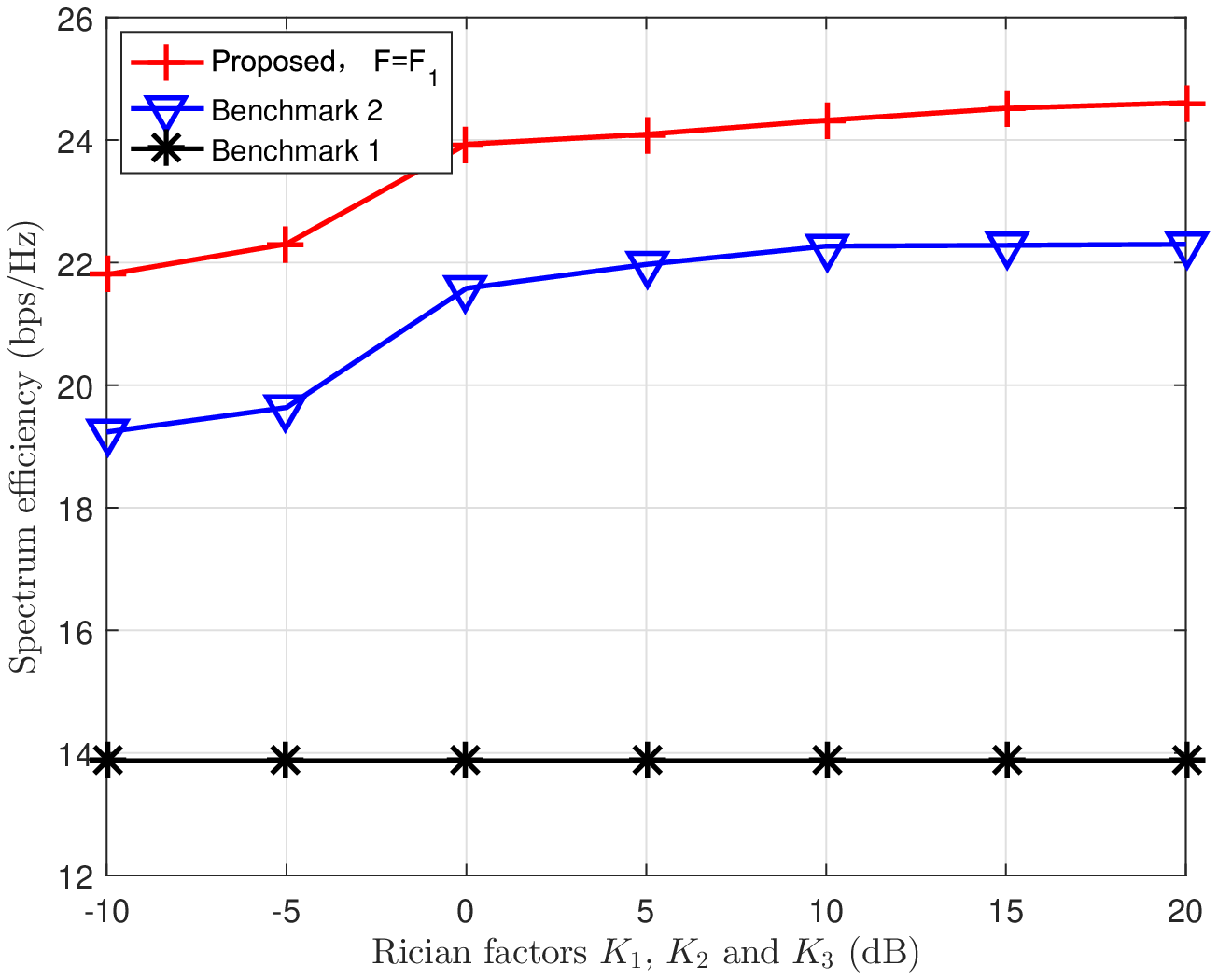}
\caption{SE versus Rician factors $K_1$, $K_2$ and $K_3$.}
\label{fig:FigSim11}
\end{minipage}%
\begin{minipage}[t]{0.5\linewidth}
\centering
\includegraphics[width=1\textwidth]{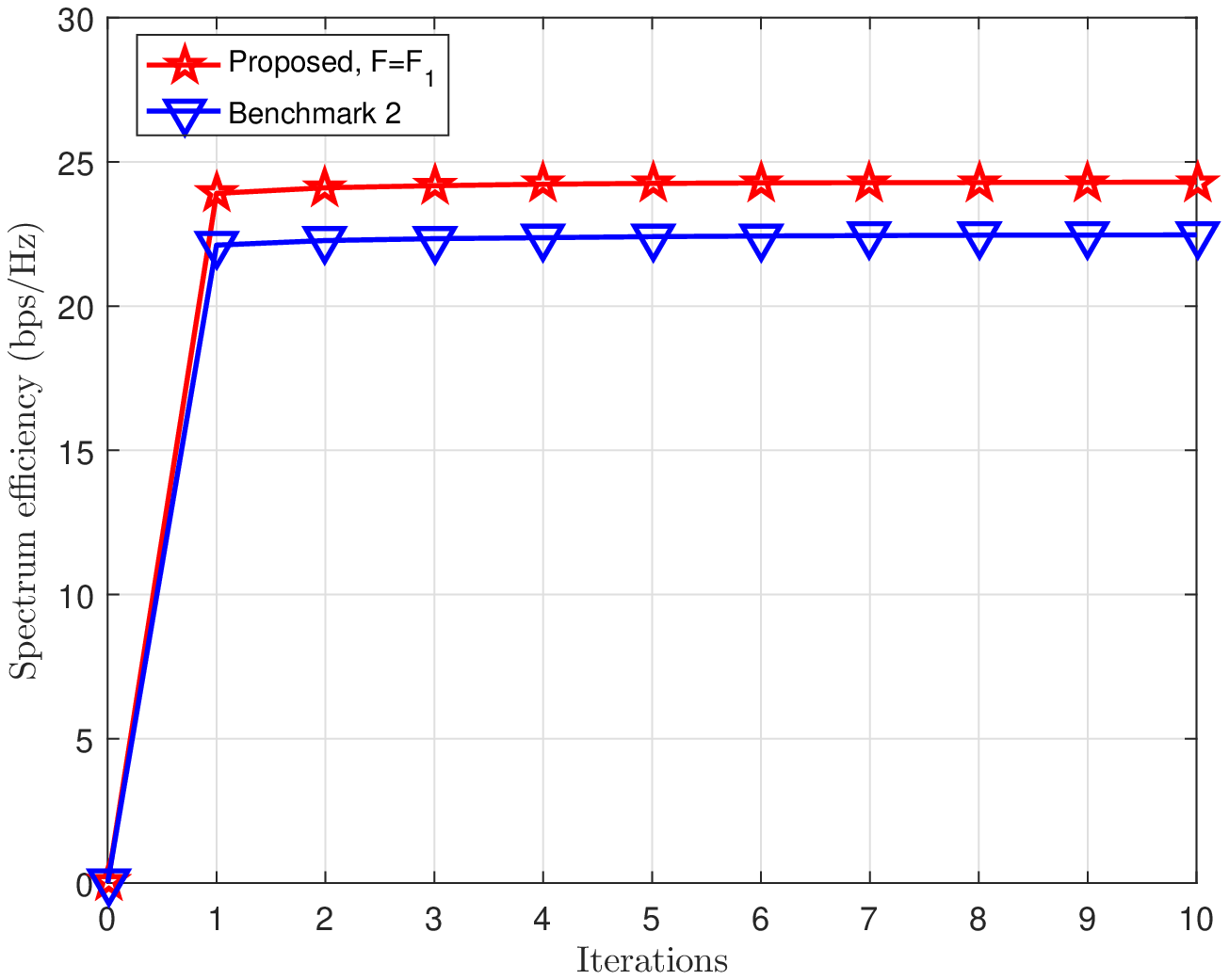}
\caption{Convergence of Algorithm \ref{AlgorithmP1} for (P1).}
\label{fig:FigSim9}
\end{minipage}
\vspace{-0.4cm}
\end{figure}

Fig.~\ref{fig:FigSim8} plots the SE versus the CUs' minimum rate requirement $R_{\min}^c$. The SE decreases as $R_{\min}^c$ increases, which reveals the rate tradeoff between D2D links and cellular links.  The proposed design achieves significant SE enhancement by introducing an RIS as compared to the first benchmark for $R_{\min}^c\leq 1.5$ bps/Hz. For $R_{\min}^c\geq 1.5$ bps/Hz, the first benchmark achieves higher SE as compared to the proposed design. The reason is that both the pairing scheme and the benefits introduced by the RIS affect the SE performance. When $R_{\min}^c$ is relatively small, the proposed design with a suboptimal pairing scheme is able to obtain a feasible solution, the performance enhancement comes from the well-designed RIS; in contrast, when $R_{\min}^c$ is too large, the formulated problem is not likely to be solved under a suboptimal pairing scheme, which results into worse performance of the proposed design as compared to the first benchmark with ideal pairing scheme.
\begin{figure} [!t]
	\centering	\includegraphics[width=.7\columnwidth]{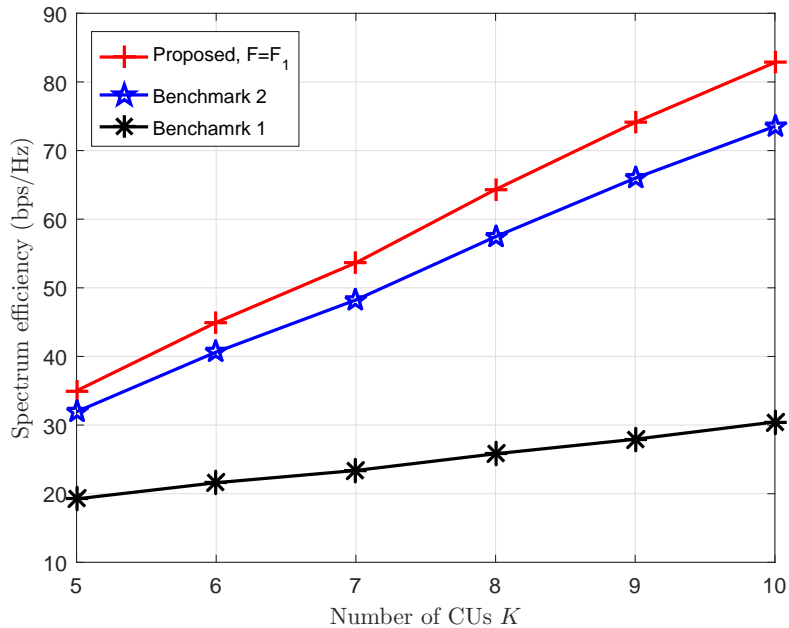}
	\caption{SE versus number of CUs $K$.} \label{fig:FigSim10}
\vspace{-0.4cm}
\end{figure}

Fig.~\ref{fig:FigSim11} plots the the SE versus the Rician factors $K_1$, $K_2$ and $K_3$.  It is observed that the SE of the proposed design increases as Rician factors increase, since the increase of $K_1, \ K_2$ and $K_3$ indicates that the LoS components between users and RIS are enhanced, thus reduces the depth of signal fading. The SE of first benchmark almost remains unchanged. Furthermore, the slope of the curve decreases with the increase of Rician factors. When Rician factors are small, NLoS path is dominant and the increase of Rician factors means the increase of LoS-path strength; when Rician factors are relatively large, LoS path is absolutely dominant thus the increase of Rician factors has less effect on the channel strength.

Fig.~\ref{fig:FigSim9} plots the average convergence performance of the proposed Algorithm \ref{AlgorithmP1}. We observe that the proposed scheme takes about five iterations to converge. The converged average SE is  24.34 bps/Hz and 22.482 bps/Hz for the proposed design and the second benchmark, respectively. Thus, the convergence speed of Algorithm \ref{AlgorithmP1} is fast.

Fig.~\ref{fig:FigSim10} plots the SE versus the number of CUs. We set the coordinates of added CUs as (99, 44), (122, 174), (138, 162), (74, 121), (56, 112) and (149, 78), respectively. The proposed design significantly outperforms the first benchmark  when the number of CUs is relatively large. For example, the SE of the proposed design is $81.82\%$ and $172\%$ higher than that of the first benchmark  when the number of CUs is 5 and 10 respectively. The reason is explained as follows. With the number of CUs increases, the SE performance for the first benchmark only benefits from the potential better pairing and the rate of added CUs, while the SE performance for the proposed design also benefits from the enhancement of the RIS.

\subsection{Simulation Analyses for EE Maximization}
\begin{figure}[t!]
\begin{minipage}[t]{0.5\linewidth}
\centering
\includegraphics[width=1\textwidth]{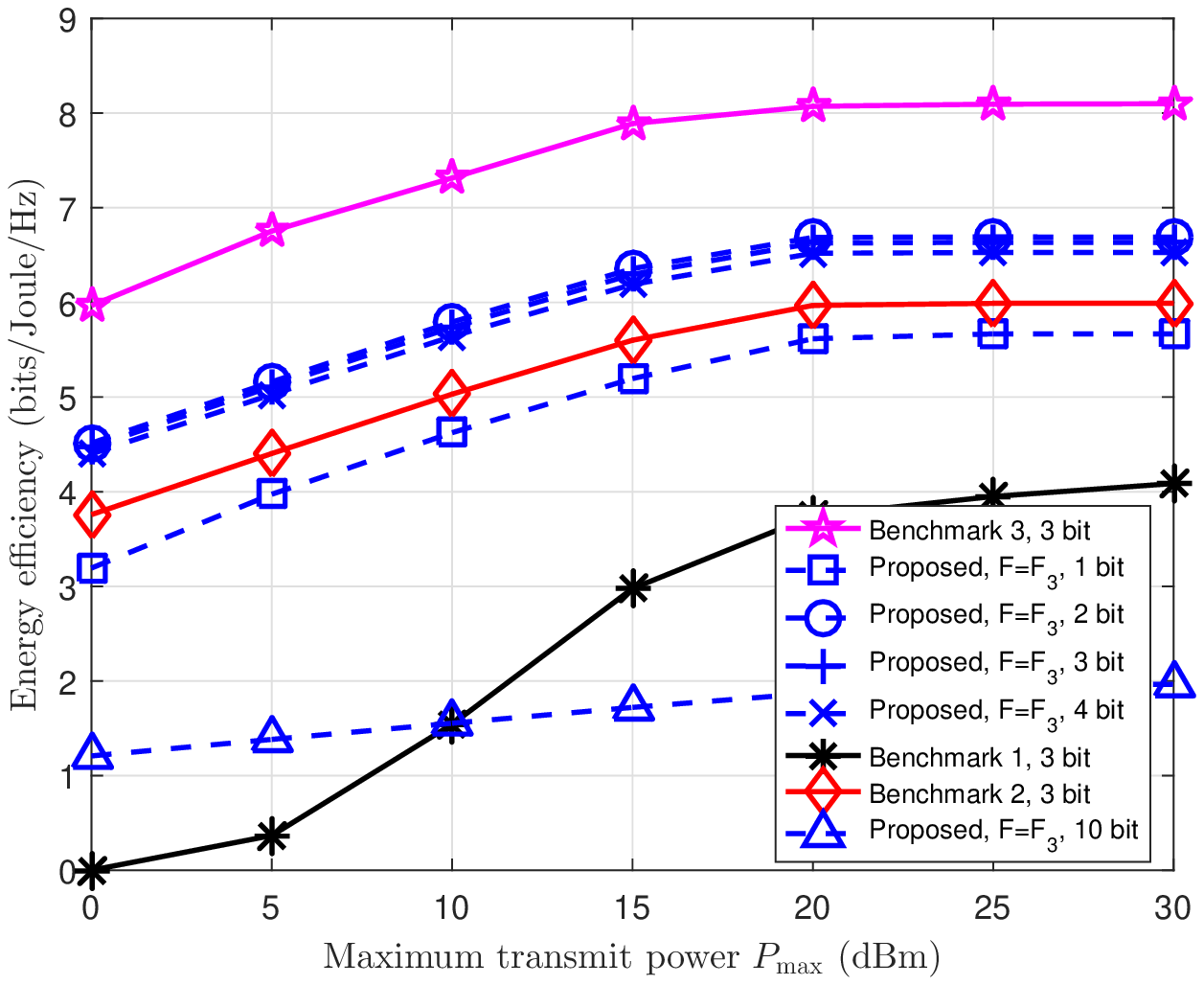}
\caption{EE versus D2D TXs' maximum trans\\-mit power $P_{\max}$.}
\label{fig:FigSim12}
\end{minipage}%
\begin{minipage}[t]{0.5\linewidth}
\centering
\includegraphics[width=1\textwidth]{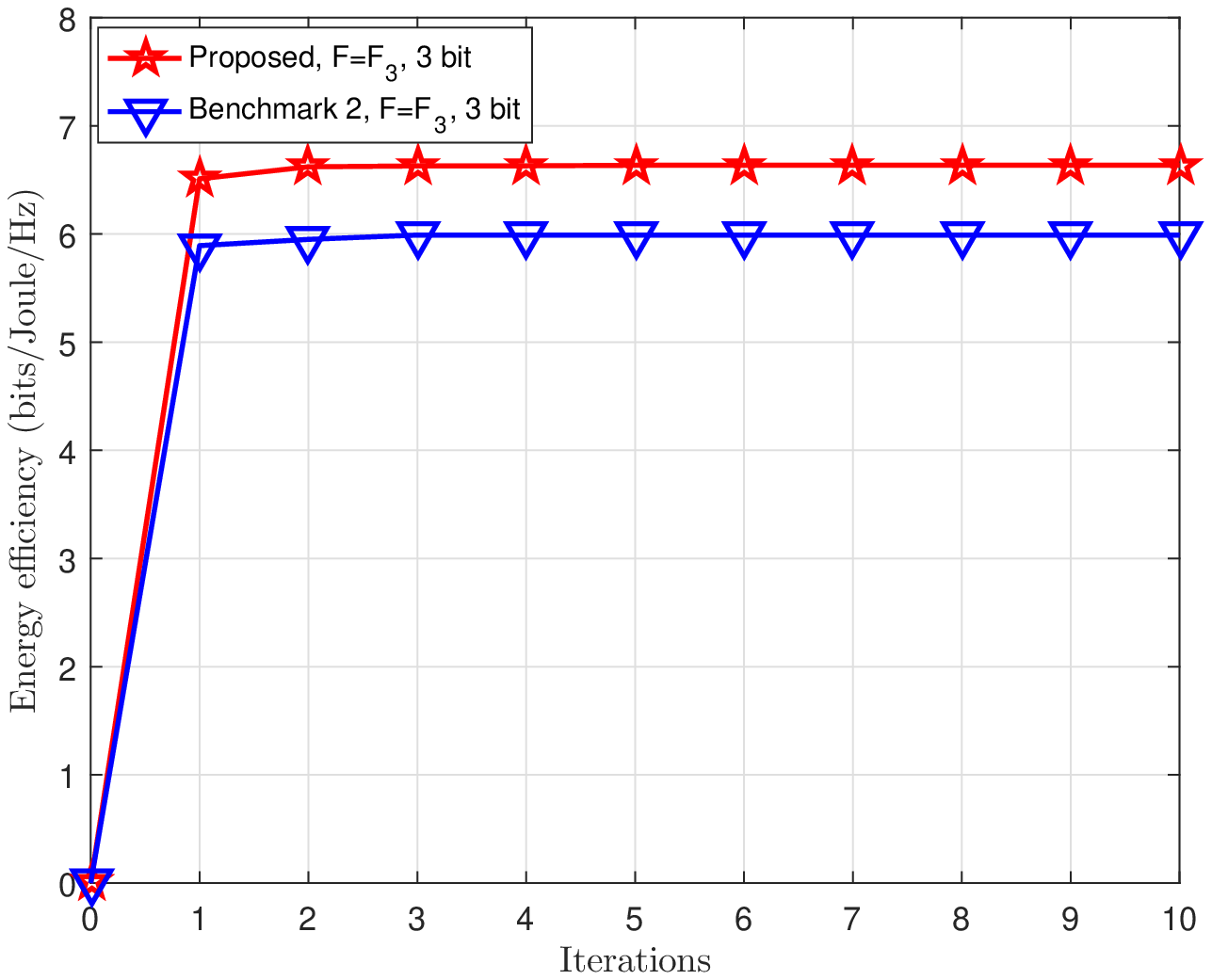}
\caption{Convergence of Dinkelbach-based Algorithm 2 for (P2).}
\label{fig:FigSim13}
\end{minipage}
\vspace{-0.4cm}
\end{figure}

In this subsection, we evaluate the EE performance. As in \cite{samithTCOM2020}, we take the diode SMV1231-079 with inverse current $i_{\sf r}$ less than 20 nA, and estimate the power $P_v(B)$ for different $B$'s. We take the typical Xilinx Spartan-7 FPGA for consideration, and use its typical power $P_{\sf FPGA}=1.188$W. The average power consumption of each reflecting element for different quantization bit $B$ with reflecting-element number $M$ is given in Table~\ref{table2}. In the simulations, we set $M=500$, $R_{\min}^d=R_{\min}^c$=0.55 bps/Hz, $P_{\sf r}=24$ dBm, and $f_s=10$ KHz. The bias voltages of varactor diode are chosen in the range of [0.12, 8.83]. Other settings remain unchanged as in subsection VI-C.
\begin{table}[h]
  \caption{Power consumption (mW) of each reflecting element.} \label{table2}
  \footnotesize
  \centering
  \begin{tabular}{c|c|c|c|c|c|c|c|c|c|c}
  \hline
  \diagbox{$M$}{$B$}&1&2&3&4&5&6&7&8&9&10\\\hline
  200& 5.970  & 6.000   & 6.060  & 6.180   & 6.420  & 6.900   & 7.860  & 9.780  & 13.620  & 21.300  \\ \hline                         500& 2.406  & 2.436   & 2.496  & 2.616   & 2.856  & 3.336   & 4.296  & 6.216  & 10.056  & 17.736  \\ \hline                                 1000& 1.218 & 1.248   & 1.308  & 1.428   & 1.668  & 2.148   & 3.108  & 5.028  & 8.868  & 16.548  \\ \hline
\end{tabular}
\end{table}

Fig.~\ref{fig:FigSim12} plots the EE versus the maximum transmit power $P_{\max}$. It is observed that EE increases as $P_{\max}$ increases first, then almost remains unchanged for large $P_{\max}$, since the increment of SE is not as fast as the increment of consumed power. The proposed design significantly outperforms the first benchmark with 3-bit phase shifts. As the number of phase-shift quantization bits $B$ increases, the EE increases first and then decreases slowly, with the optimal EE at $B=2$. The reason is that despite the increase of $B$ makes the setting of reflecting coefficients more accurate, but it results into higher power consumption simultaneously. When $B>1$, the SE improvement from RIS is not enough to compensate for the increase of power consumption as $B$ increases. In addition, the proposed design outperforms the second benchmark, due to the benefit of passive beamforming. Moreover, the proposed design suffers from slight EE performance degradation, but outperforms this benchmark in terms of computational complexity.

Fig.~\ref{fig:FigSim13} plots the average convergence performance of Dinkelbach-based Algorithm 2 for solving EE-maximization problem (P2). We observe that the proposed design takes about five iterations to converge. The converged average EE is  6.35 bps/Joule/Hz and 5.99 bps/Joule/Hz for the proposed design and the second benchmark, respectively. Thus, the convergence of Algorithm 2 is fast.

\vspace{-0.2cm}
\section{Conclusion}\label{conslusion}
This paper has studied an RIS-empowered underlaying D2D communication network. The overall network's spectrum efficiency (SE) and energy efficiency (EE) are maximized, respectively, by jointly optimizing the resource reuse indicators, the transmit power and the passive beamforming. First, an efficient relative-channel-strength based user-pairing scheme with low complexity is proposed to determine the resource reuse indicators. Then, the transmit power and the passive beamforming are jointly optimized for maximizing SE and EE, respectively, by utilizing the proposed alternating-optimization based iterative algorithms. Numerical results show that the proposed design achieves significant performance enhancement compared to traditional underlaying D2D network without RIS, and suffers from slight performance degradation compared to RIS-empowered underlying D2D with ideal user-pairing. This work can be extended to practical and complex scenarios such as multi-antenna BS/users, multiple RISs, and imperfect/partial channel state information.

\vspace{-0.2cm}
\bibliography{IEEEabrv,reference20190808}

\begin{thebibliography}{10}
\providecommand{\url}[1]{#1}
\csname url@samestyle\endcsname
\providecommand{\newblock}{\relax}
\providecommand{\bibinfo}[2]{#2}
\providecommand{\BIBentrySTDinterwordspacing}{\spaceskip=0pt\relax}
\providecommand{\BIBentryALTinterwordstretchfactor}{4}
\providecommand{\BIBentryALTinterwordspacing}{\spaceskip=\fontdimen2\font plus
\BIBentryALTinterwordstretchfactor\fontdimen3\font minus
  \fontdimen4\font\relax}
\providecommand{\BIBforeignlanguage}[2]{{%
\expandafter\ifx\csname l@#1\endcsname\relax
\typeout{** WARNING: IEEEtran.bst: No hyphenation pattern has been}%
\typeout{** loaded for the language `#1'. Using the pattern for}%
\typeout{** the default language instead.}%
\else
\language=\csname l@#1\endcsname
\fi
#2}}
\providecommand{\BIBdecl}{\relax}
\BIBdecl

\bibitem{Jameel2018}
F.~{Jameel}, Z.~{Hamid}, F.~{Jabeen}, S.~{Zeadally}, and M.~A. {Javed}, ``A
  survey of device-to-device communications: Research issues and challenges,''
  \emph{IEEE Commun. Sur. \& Tut.}, vol.~20, no.~3, pp. 2133--2168, 2018.

\bibitem{Ghosh2019}
A.~{Ghosh}, A.~{Maeder}, M.~{Baker}, and D.~{Chandramouli}, ``{5G} evolution: A
  view on {5G} cellular technology beyond {3GPP} release 15,'' \emph{IEEE
  Access}, vol.~7, pp. 127\,639--127\,651, 2019.

\bibitem{Naderializadeh2014}
N.~{Naderializadeh} and A.~S. {Avestimehr}, ``{ITLinQ}: A new approach for
  spectrum sharing in device-to-device communication systems,'' \emph{IEEE J.
  Sel. Areas Commun.}, vol.~32, no.~6, pp. 1139--1151, 2014.

\bibitem{DSABook2019Liang}
Y.-C. Liang, \emph{Dynamic Spectrum Management: From Cognitive Radio to
  Blockchain and Artificial Intelligence}, ser. Signals and Communication
  Technology.\hskip 1em plus 0.5em minus 0.4em\relax Singapore: Springer, 2020.

\bibitem{Sun2015D2D}
H.~{Sun}, M.~{Wildemeersch}, M.~{Sheng}, and T.~Q.~S. {Quek}, ``{D2D} enhanced
  heterogeneous cellular networks with dynamic {TDD},'' \emph{IEEE Trans.
  Wireless Commun.}, vol.~14, no.~8, pp. 4204--4218, 2015.

\bibitem{Wang2018}
J.~{Wang}, Y.~{Huang}, S.~{Jin}, R.~{Schober}, X.~{You}, and C.~{Zhao},
  ``Resource management for device-to-device communication: A physical layer
  security perspective,'' \emph{IEEE J. Sel. Areas Commun.}, vol.~36, no.~4,
  pp. 946--960, 2018.

\bibitem{Yang2016}
K.~{Yang}, S.~{Martin}, C.~{Xing}, J.~{Wu}, and R.~{Fan}, ``Energy-efficient
  power control for device-to-device communications,'' \emph{IEEE J. Sel. Areas
  Commun.}, vol.~34, no.~12, pp. 3208--3220, 2016.

\bibitem{Chen2018}
Y.~{Chen}, B.~{Ai}, Y.~{Niu}, K.~{Guan}, and Z.~{Han}, ``Resource allocation
  for device-to-device communications underlaying heterogeneous cellular
  networks using coalitional games,'' \emph{IEEE Trans. Wireless Commun.},
  vol.~17, no.~6, pp. 4163--4176, 2018.

\bibitem{liaskos2018new}
C.~Liaskos, S.~Nie, A.~Tsioliaridou, A.~Pitsillides, S.~Ioannidis, and
  I.~Akyildiz, ``A new wireless communication paradigm through
  software-controlled metasurfaces,'' \emph{IEEE Commun. Mag.}, vol.~56, no.~9,
  pp. 162--169, 2018.

\bibitem{di2019smart}
M.~Di~Renzo, M.~Debbah, and et. {\emph{al.}}, ``Smart radio environments
  empowered by reconfigurable {AI} meta-surfaces: an idea whose time has
  come,'' \emph{EURASIP J. Wireless Commun. Netw.}, vol. 2019, no.~1, pp.
  1--20, 2019.

\bibitem{Liang2019}
Y.-C. {Liang}, R.~{Long}, Q.~{Zhang}, J.~{Chen}, H.~V. {Cheng}, and H.~{Guo},
  ``Large intelligent surface/antennas ({LISA}): Making reflective radios
  smart,'' \emph{J. Commun. Inf. Netw.}, vol.~4, no.~2, pp. 40--50, 2019.

\bibitem{huang2020holographic}
C.~Huang, S.~Hu, G.~C. Alexandropoulos, A.~Zappone, C.~Yuen, R.~Zhang,
  M.~Di~Renzo, and M.~Debbah, ``Holographic {MIMO} surfaces for {6G} wireless
  networks: Opportunities, challenges, and trends,'' \emph{IEEE Wireless
  Commun., \emph{Doi: 10.1109/MWC.001.1900534}}, 2020.

\bibitem{yu2011resource}
C.-H. Yu, K.~Doppler, C.~B. Ribeiro, and O.~Tirkkonen, ``Resource sharing
  optimization for device-to-device communication underlaying cellular
  networks,'' \emph{IEEE Trans. Wireless Commun.}, vol.~10, no.~8, pp.
  2752--2763, 2011.

\bibitem{feng2013device}
D.~Feng, L.~Lu, Y.~Yuan-Wu, G.~Y. Li, G.~Feng, and S.~Li, ``Device-to-device
  communications underlaying cellular networks,'' \emph{IEEE Trans. Commun.},
  vol.~61, no.~8, pp. 3541--3551, 2013.

\bibitem{Kai2019}
Y.~{Kai}, J.~{Wang}, H.~{Zhu}, and J.~{Wang}, ``Resource allocation and
  performance analysis of cellular-assisted {OFDMA} device-to-device
  communications,'' \emph{IEEE Trans. Wireless Commun.}, vol.~18, no.~1, pp.
  416--431, 2019.

\bibitem{MirzaJoint2018}
J.~{Mirza}, G.~{Zheng}, K.~{Wong}, and S.~{Saleem}, ``Joint beamforming and
  power optimization for {D2D} underlaying cellular networks,'' \emph{IEEE
  Trans. Veh. Technol.}, vol.~67, no.~9, pp. 8324--8335, 2018.

\bibitem{cui2019spatial}
W.~Cui, K.~Shen, and W.~Yu, ``Spatial deep learning for wireless scheduling,''
  \emph{IEEE J. Sel. Areas Commun.}, vol.~37, no.~6, pp. 1248--1261, 2019.

\bibitem{Wang2015}
F.~{Wang}, C.~{Xu}, L.~{Song}, and Z.~{Han}, ``Energy-efficient resource
  allocation for device-to-device underlay communication,'' \emph{IEEE Trans.
  Wireless Commun.}, vol.~14, no.~4, pp. 2082--2092, 2015.

\bibitem{Feng2015}
D.~{Feng}, G.~{Yu}, C.~{Xiong}, Y.~{Yuan-Wu}, G.~Y. {Li}, G.~{Feng}, and
  S.~{Li}, ``Mode switching for energy-efficient device-to-device
  communications in cellular networks,'' \emph{IEEE Tran. Wireless Commun.},
  vol.~14, no.~12, pp. 6993--7003, 2015.

\bibitem{Chun2017}
Y.~J. {Chun}, S.~L. {Cotton}, H.~S. {Dhillon}, A.~{Ghrayeb}, and M.~O. {Hasna},
  ``A stochastic geometric analysis of device-to-device communications
  operating over generalized fading channels,'' \emph{IEEE Trans. Wireless
  Commun.}, vol.~16, no.~7, pp. 4151--4165, 2017.

\bibitem{WangD2D2019}
Y.~{Wang}, Y.-P. {Hong}, and W.~{Chen}, ``Dynamic transmission policy for
  multi-pair cooperative device-to-device communication with
  block-diagonalization precoding,'' \emph{IEEE Trans. Wireless Commun.},
  vol.~18, no.~6, pp. 3034--3048, 2019.

\bibitem{Penda2019}
D.~D. {Penda}, R.~{Wichman}, T.~{Charalambous}, G.~{Fodor}, and M.~{Johansson},
  ``A distributed mode selection scheme for full-duplex device-to-device
  communication,'' \emph{IEEE Trans. Veh. Technol.}, vol.~68, no.~10, pp.
  10\,267--10\,271, 2019.

\bibitem{Tang2017}
A.~{Tang}, X.~{Wang}, and C.~{Zhang}, ``Cooperative full duplex device to
  device communication underlaying cellular networks,'' \emph{IEEE Trans.
  Wireless Commun.}, vol.~16, no.~12, pp. 7800--7815, 2017.

\bibitem{Mozaffari2016}
M.~{Mozaffari}, W.~{Saad}, M.~{Bennis}, and M.~{Debbah}, ``Unmanned aerial
  vehicle with underlaid device-to-device communications: Performance and
  tradeoffs,'' \emph{IEEE Trans. Wireless Commun.}, vol.~15, no.~6, pp.
  3949--3963, 2016.

\bibitem{Di2020smart}
M.~Di~Renzo, A.~Zappone, M.~Debbah, M.-S. Alouini, C.~Yuen, J.~de~Rosny, and
  S.~Tretyakov, ``Smart radio environments empowered by reconfigurable
  intelligent surfaces: How it works, state of research, and road ahead,''
  \emph{arXiv preprint arXiv:2004.09352}, 2020.

\bibitem{EmilIRSDFWCL2020}
E.~{Bj\"{o}rnson}, O.~{\"{O}zdogan}, and E.~G. {Larsson}, ``Intelligent
  reflecting surface versus decode-and-forward: How large surfaces are needed
  to beat relaying?'' \emph{IEEE Wireless Communications Letters}, vol.~9,
  no.~2, pp. 244--248, 2020.

\bibitem{YangLiangZhangPeiTCOM17}
G.~Yang, Y.-C. Liang, R.~Zhang, and Y.~Pei, ``Modulation in the air:
  Backscatter communication over ambient {OFDM} carrier,'' \emph{{IEEE} Trans.
  Commun.}, vol.~66, no.~3, pp. 1219--1233, Mar. 2018.

\bibitem{YangLiangZhangIoTJ18}
G.~Yang, Q.~Zhang, and Y.-C. Liang, ``Cooperative ambient backscatter
  communications for green {Internet-of-Things},'' \emph{{IEEE} Internet of
  Things J.}, vol.~5, no.~2, pp. 1116--1130, Apr. 2018.

\bibitem{Guo2020}
H.~{Guo}, Y.-C. {Liang}, J.~{Chen}, and E.~G. {Larsson}, ``Weighted sum-rate
  maximization for reconfigurable intelligent surface aided wireless
  networks,'' \emph{IEEE Tran. Wireless Commun.}, vol.~19, no.~5, pp.
  3064--3076, 2020.

\bibitem{Huang2020}
C.~{Huang}, R.~{Mo}, and C.~{Yuen}, ``Reconfigurable intelligent surface
  assisted multiuser miso systems exploiting deep reinforcement learning,''
  \emph{IEEE J. Sel. Areas Commun. (Early Access)}, 2020.

\bibitem{Jung2019}
M.~{Jung}, W.~{Saad}, Y.~{Jang}, G.~{Kong}, and S.~{Choi}, ``Reliability
  analysis of large intelligent surfaces ({LIS}s): Rate distribution and outage
  probability,'' \emph{IEEE Wireless Commun. Lett.}, vol.~8, no.~6, pp.
  1662--1666, 2019.

\bibitem{YangNOMA2020}
G.~{Yang}, X.~{Xu}, and Y.-C. {Liang}, ``Intelligent reflecting surface
  assisted non-orthogonal multiple access,'' in \emph{IEEE Wireless Commun.
  Netw. Conf. (WCNC), Seoul, Korea (South)}, 2020.

\bibitem{8741198}
C.~{Huang}, A.~{Zappone}, G.~C. {Alexandropoulos}, M.~{Debbah}, and C.~{Yuen},
  ``Reconfigurable intelligent surfaces for energy efficiency in wireless
  communication,'' \emph{IEEE Trans. Wireless Commun.}, vol.~18, no.~8, pp.
  4157--4170, 2019.

\bibitem{JChen2019}
J.~{Chen}, Y.-C. {Liang}, Y.~{Pei}, and H.~{Guo}, ``Intelligent reflecting
  surface: A programmable wireless environment for physical layer security,''
  \emph{IEEE Access}, vol.~7, pp. 82\,599--82\,612, 2019.

\bibitem{RISD2DsubmittedGC2020}
G.~{Yang}, Y.~{Liao}, Y.-C. {Liang}, and O.~{Tirkkonen}, ``Reconfigurable
  intelligent surface empowered underlaying device-to-device communication,''
  submitted to 2020 IEEE Global Communications Conference, available:
  \url{https://arxiv.org/abs/2006.02103}.

\bibitem{RamezanijointD2D}
A.~{Ramezani-Kebrya}, M.~{Dong}, B.~{Liang}, G.~{Boudreau}, and S.~H.
  {Seyedmehdi}, ``Joint power optimization for device-to-device communication
  in cellular networks with interference control,'' \emph{IEEE Trans. Wireless
  Commun.}, vol.~16, no.~8, pp. 5131--5146, 2017.

\bibitem{beck2010sequential}
A.~Beck, A.~Ben-Tal, and L.~Tetruashvili, ``A sequential parametric convex
  approximation method with applications to nonconvex truss topology design
  problems,'' \emph{J. Global Opt.}, vol.~47, no.~1, pp. 29--51, 2010.

\bibitem{grant2008cvx}
M.~Grant, S.~Boyd, and Y.~Ye, ``{CVX}: Matlab software for disciplined convex
  programming,'' 2008.

\bibitem{shen2018fractional}
K.~Shen and W.~Yu, ``Fractional programming for communication systems¡ªpart
  {I}: power control and beamforming,'' \emph{IEEE Trans. Signal Processing},
  vol.~66, no.~10, pp. 2616--2630, 2018.

\bibitem{Bharadia2011}
D.~{Bharadia}, G.~{Bansal}, P.~{Kaligineedi}, and V.~K. {Bhargava}, ``Relay and
  power allocation schemes for {OFDM}-based cognitive radio systems,''
  \emph{IEEE Trans. Wireless Commun.}, vol.~10, no.~9, pp. 2812--2817, 2011.

\bibitem{hassanien2008robust}
A.~Hassanien, S.~A. Vorobyov, and K.~M. Wong, ``Robust adaptive beamforming
  using sequential quadratic programming: An iterative solution to the mismatch
  problem,'' \emph{IEEE Signal Processing Lett.}, vol.~15, pp. 733--736, 2008.

\bibitem{Ribeiro2018}
L.~N. {Ribeiro}, S.~{Schwarz}, M.~{Rupp}, and A.~L.~F. {de Almeida}, ``Energy
  efficiency of mmwave massive {MIMO} precoding with low-resolution {DAC}s,''
  \emph{IEEE J. Sel. Topics in Sig. Proc.}, vol.~12, no.~2, pp. 298--312, 2018.

\bibitem{crouzeix1991algorithms}
J.-P. Crouzeix and J.~A. Ferland, ``Algorithms for generalized fractional
  programming,'' \emph{Math. Programm.}, vol.~52, no. 1-3, pp. 191--207, 1991.

\bibitem{samithTCOM2020}
S.~Abeywickrama, R.~Zhang, Q.~Wu, and C.~Yuen, ``Intelligent reflecting
  surface: Practical phase shift model and beamforming optimization,'' \emph{to
  appear in IEEE Trans. Commun., available at arxiv.org/abs/2002.10112}, 2020.

\end{thebibliography}
\bibliographystyle{IEEEtran}

\end{document}